\documentclass[12pt,english]{article}
\usepackage[T1]{fontenc}
\usepackage[latin9]{inputenc}
\usepackage{geometry}
\usepackage{color}
\usepackage{babel}
\usepackage{amsmath}
\usepackage{amsthm}
\usepackage{amssymb}
\usepackage{graphicx}
\usepackage{setspace}
\usepackage[authoryear]{natbib}
\usepackage[unicode=true,pdfusetitle,
 bookmarks=true,bookmarksnumbered=false,bookmarksopen=false,
 breaklinks=false,pdfborder={0 0 1},backref=false,colorlinks=true]
 {hyperref}
\hypersetup{
 citecolor=blue,  urlcolor=cyan}

\makeatletter

\providecommand{\tabularnewline}{\\}

\theoremstyle{plain}
    \ifx\thechapter\undefined
      \newtheorem{prop}{\protect\propositionname}
    \else
      \newtheorem{prop}{\protect\propositionname}[chapter]
    \fi
\theoremstyle{plain}
    \ifx\thechapter\undefined
      \newtheorem{lem}{\protect\lemmaname}
    \else
      \newtheorem{lem}{\protect\lemmaname}[chapter]
    \fi

\newcommand\independent{\protect\mathpalette{\protect\independenT}{\perp}}
\def\independenT#1#2{\mathrel{\rlap{$#1#2$}\mkern2mu{#1#2}}}
\definecolor{green}{rgb}{0.05, 0.5, 0.06}
\usepackage[superscript,biblabel]{cite}
\usepackage{babel}

\makeatother

\providecommand{\lemmaname}{Lemma}
\providecommand{\propositionname}{Proposition}

\begin{document}
\date{}
\title{\doublespacing{}\noindent {\Huge{}Efficient semiparametric estimation of marginal
treatment effects with genetic instrumental variables}\thanks{Ashish Patel (\protect\href{mailto:ashish.patel@mrc-bsu.cam.ac.uk}{ashish.patel@mrc-bsu.cam.ac.uk});
Francis J.\@ DiTraglia (\protect\href{mailto:francis.ditraglia@economics.ox.ac.uk}{francis.ditraglia@economics.ox.ac.uk});
Stephen Burgess (\protect\href{http://sb452@medschl.cam.ac.uk}{sb452@medschl.cam.ac.uk}).
We thank Jack Bowden for helpful discussions. }}
\author{Ashish Patel{\small{} }\textsuperscript{{\small{}a}}{\small{}, }Francis
J.\@ DiTraglia{\small{} }\textsuperscript{{\small{}b}}{\small{},}
\& Stephen Burgess{\small{} }\textsuperscript{{\small{}a,c}}}
\maketitle
\noindent \begin{center}
{\small{}\vskip -2em}\textsuperscript{{\small{}a}}{\small{} MRC
Biostatistics Unit, University of Cambridge}\\
{\small{}}\textsuperscript{{\small{}b}}{\small{} Department of Economics,
University of Oxford}\\
{\small{}}\textsuperscript{{\small{}c}}{\small{} Cardiovascular
Epidemiology Unit, University of Cambridge}{\small\par}
\par\end{center}
\begin{abstract}
Alcohol misuse is a key target of public health strategies aimed at
reducing cardiovascular risk. The effect of excessive alcohol consumption
on blood pressure may vary systematically with individuals' unobserved
propensity to engage in heavy drinking, complicating causal inference
with observational data. The marginal treatment effects framework
uses an instrumental variable for treatment choice (excessive alcohol
consumption) to study how selection into treatment is linked with
the treatment effect. We explore the use of a genetic instrument within
this framework, which is challenging because genetic compliers (individuals
for whom a change in the instrument changes their treatment choice)
are likely to be a small proportion of the overall sample. This can
lead to greater sampling uncertainty in the tails of the propensity
score distribution, i.e.\@, the conditional probability of choosing
treatment, and in turn poor estimation of causal estimands that measure
heterogeneous treatment effects. We show that the use of efficient
influence functions of target estimands improves estimation in terms
of robustness to sampling uncertainty in nonparametrically estimated
propensity scores. We find evidence of \textsl{reverse selection on
gains}: individuals most prone to excessive alcohol consumption experience
larger adverse effects on blood pressure. 
\end{abstract}
\newpage{}

\section{Introduction}

Excessive alcohol consumption is a modifiable behaviour that is responsible
for an estimated 2.6 million deaths per year globally according to
World Health Organization estimates \citep{Anderson2023}. Alcohol
misuse features prominently in public health recommendations aimed
at reducing cardiovascular risk, with clinical trial evidence establishing
a causal link between alcohol consumption and higher blood pressure
levels \citep{Tasnim2020}. 

Individuals differ substantially in their alcohol drinking habits
and responsiveness to health information \citep{DiMatteo2004}, suggesting
the effects of excessive alcohol consumption could also be heterogeneous
and closely tied to drinking behaviour. For example, those with high
blood pressure may be advised to lower alcohol intake which may attenuate
any positive association. Understanding this relationship between
the effects of excessive alcohol consumption (the \textsl{treatment
effect}) and the choice to drink excessively (the \textsl{treatment
choice}) is crucial for the design of effective public health interventions
\citep{Heckman2001}.

In practice, we cannot randomly assign individuals to drink alcohol
excessively, which leaves the concern that an observed association
between treatment choice and the outcome (blood pressure) does not
reflect only the causal effect of treatment but may instead be distorted
by unmeasured factors that correlate the two, such as dietary, smoking,
and exercise habits \citep{DiFederico2023}. The method of instrumental
variables (IVs) is an important tool in applied research to infer
causal effects by leveraging sources of variation in treatment choice
that are unrelated to unobserved confounders. 

In genetic epidemiology, Mendelian randomization uses genetic variants
as IVs, and aims to exploit the quasi-random allocation of genetic
variants at conception. In particular, genetic variants may be less
prone to issues of reverse causation and confounding \citep{Smith2003}.
The approach is most plausible when the biological mechanism by which
the genetic instrument influences treatment choice is known \citep{Larsson2023}.
Our instrument is a genetic risk score for alcohol consumption in
which a genetic variant from the \textsl{ADH1B} gene region receives
the largest weight, and is known to affect the ability to metabolize
alcohol and thus the risk of alcohol dependence \citep{Rietschel2013}. 

Recent years have seen a sharp rise in the number of Mendelian randomization
studies being performed \citep{Hemani2025}, with the vast majority
of studies estimating a simple linear IV model, even when the treatment
choice is binary \citep{Wu2024}. A large literature has focused on
clarifying the interpretation of IV estimators in simple linear models
that assume a constant treatment effect, when in fact there is unobserved
heterogeneity linking treatment effects with treatment choice. \citet{Imbens1994}
discussed conditions under which the IV estimand can be interpreted
as an average treatment effect for \textsl{compliers}, a subgroup
of the population for whom a change in the instrument changes their
treatment choice. 

There are two reasons why the IV estimand is unlikely to be useful
for Mendelian randomization analyses under unobserved heterogeneity
in treatment effects. First, genetics are likely to explain a low
proportion of variability in treatment choice, and so genetic compliers
may be only a small, and potentially unrepresentative, proportion
of the population of interest \citep{Burgess2018}. Second, the IV
estimand may not be interesting from a policy perspective because
it is not possible to directly manipulate the instrument (genetic
variants). 

A closely related strategy that directly estimates the relationship
between treatment effects and treatment choice is the marginal treatment
effects (MTE, \citealt{Heckman1999,Heckman2005}) framework. The MTE
estimand measures how treatment effects vary with a continuous latent
variable representing individual-specific reluctance to take treatment,
and it can also be reweighted to derive conventional estimands such
as the average treatment effect. 

In this paper, we apply the MTE framework to study treatment effect
heterogeneity with a genetic instrumental variable. MTE estimation
relies on first stage estimation of a propensity score (defined as
the conditional probability of treatment given observed covariates
and instrument). Due to the potentially complex nature of how genetic
variants influence excessive alcohol consumption \citep{Edenberg2013},
we will not impose a parametric assumption on the propensity score.
Instead, we will study the general case where the propensity score
is estimated nonparametrically.

The conventional approach to estimation in the MTE framework uses
a simple linear regression and does not account for uncertainty in
first-stage propensity score estimation. When the instrument induces
only limited variation in treatment, the estimated propensity score
is concentrated over a narrow range, leading to poor estimation of
the marginal treatment effects curve. These conventional estimates
of target estimands will \textsl{not }be robust nor efficient. 

Thus, for improved estimation in the MTE framework with genetic instruments,
we derive semiparametrically efficient estimates of target estimands
under plug-in nonparametric propensity scores. Importantly, these
efficient estimates have low variance, and are robust to sampling
uncertainty in the nonparametrically estimated propensity score. We
illustrate the benefits of this robustness property in desensitizing
estimation to both weak instruments and the choice of tuning parameters
for nonparametric propensity score estimation. In our linear model,
calculating efficient estimates and their standard errors is computationally
straightforward. 

Proofs of technical results are provided in Supplementary Material.
An R package to implement the methods discussed in this paper is provided
at \href{http://github.com/ash-res/efficient.mte/}{github.com/ash-res/efficient.mte/}. 

\section{Model and Estimands }

\subsection{Model}

To estimate the causal effects of a binary treatment, the majority
of Mendelian randomization applications continue to estimate a simple
linear instrumental variable (IV) model \citep{Wu2024}. However,
when unobserved heterogeneity in treatment choice is related to treatment
effects, the standard IV estimand can be difficult to interpret \citep{Imbens1994}.
In particular, the IV estimand reflects an average treatment effect
for compliers, a small and potentially unrepresentative subgroup of
the overall population \citep{Burgess2018}. Moreover, the IV estimand
may be of limited policy relevance since the genetic instrument cannot
be altered. 

The marginal treatment effects framework offers a way to move beyond
this complier average, and to directly estimate how treatment effects
vary with individuals' propensity to choose treatment, as well as
to recover conventional estimands (such as the average treatment effect)
as simple weighted averages. 

We start by introducing the generalized \citet{Roy1951} model which
captures unobserved heterogeneity linking treatment effects with the
choice to take treatment. Let $A$ denote a binary indicator for treatment,
which equals 1 if treated, and 0 if untreated. Let $Z$ denote an
instrumental variable for treatment, and $\boldsymbol{X}$ denote
a vector of observed covariates. To relate unobserved heterogeneity
in treatment effects with treatment choice, it is useful to start
by modelling the selection process. We assume treatment choice is
described by the model 
\begin{equation}
A=I\{\pi_{{\cal P}}(\boldsymbol{X},Z)>V\}
\end{equation}

where $\pi_{{\cal P}}(\boldsymbol{X},Z)$ is an \textsl{unknown} function
of $\boldsymbol{X}$ and $Z$, but is monotone in $Z$ conditional
on $\boldsymbol{X}$. The subscript ${\cal P}$ denotes that it is
the true or population average, rather than an estimate or sample
average. The \textsl{unobserved} variable $V$ represents unobserved
heterogeneity across individuals which we can think of as describing
their reluctance to take treatment. For any fixed level of covariates
$\boldsymbol{X}$ and instrument $Z$, those individuals with high
levels of $V$ are less likely to choose treatment than those with
low levels of $V$. Since the function $\pi_{{\cal P}}(\boldsymbol{X},Z)$
is left unspecified, Equation (1) is a quite general specification. 

For each individual, let $Y_{0}$ denote their untreated potential
outcome, and $Y_{1}$ their treated potential outcome. The generalized
Roy model allows the individual-specific treatment effect $Y_{1}-Y_{0}$
to vary with $V$, which captures the dynamic that unobserved heterogeneity
in treatment effects and treatment choice are correlated. To estimate
treatment effects in this setting, we will make use of a valid instrument
$Z$ such that 
\begin{equation}
Z\,\text{is conditionally independent of }(Y_{0},Y_{1},V)\,\text{given}\,\boldsymbol{X}.
\end{equation}
In the language of IVs, Equation (2) is an \textsl{exclusion restriction}:
for any fixed level of covariates, the instrument is independent of
the potential outcomes, and so the only way it may influence the outcome
is through its effect on treatment choice. The instrument can influence
the propensity to take treatment by shifting $\pi_{{\cal P}}(\boldsymbol{X},Z)$
without changing unobserved heterogeneity $V$, and thus can be used
to capture variation in treatment choice that is unrelated to unobserved
confounding.

To analyze treatment effect heterogeneity, it is easier to think about
how the treatment effect varies with quantiles of unobserved heterogeneity
$V$, rather than its absolute values. To this end, it is standard
to use the probability integral transform to normalize $V$ as a standard
uniform random variable conditional on $\boldsymbol{X}$, so that
$\pi_{{\cal P}}(\boldsymbol{X},Z)={\cal P}(A=1\vert\boldsymbol{X},Z)$
in Equation (1) becomes the propensity score of treatment. We emphasize
that this transformation is not a modelling assumption, but a normalization
given Equations (1) and (2) if the conditional distribution of $V$
given covariates $\boldsymbol{X}$ is continuous. For our empirical
investigation, since the genetics of excessive alcohol consumption
is complex and polygenic \citep{Edenberg2013}, we will not impose
a parametric assumption on the propensity score $\pi_{{\cal P}}(\boldsymbol{X},Z)$,
but instead estimate it nonparametrically.

Next, we model the dependence of the potential outcome under each
treatment choice on unobserved reluctance to take treatment $V$.
We assume a simple linear model for conditional potential outcome
means such that 
\begin{equation}
E[Y_{a}\vert\boldsymbol{X},V]=\underset{\text{constant term}}{\underbrace{\alpha_{a}-0.5\zeta_{a}}}+\boldsymbol{\beta_{a}}^{\prime}\boldsymbol{X}+\zeta_{a}V
\end{equation}
for $a=0,1$, where $\alpha_{a}$, $\boldsymbol{\beta_{a}}$, and
$\zeta_{a}$ are unknown parameters. The unusual form of the constant
term is without loss of generality, and allows the error term $Y_{a}-E[Y_{a}\vert\boldsymbol{X}]$
for $a=0,1$ to have mean zero which simplifies derivations. Equation
(3) is the typical model used in practice where potential outcomes
can vary linearly in observed covariates $\boldsymbol{X}$ and unobserved
reluctance to take treatment $V$.

The so-called marginal treatment effect \citep*[MTE;][]{Heckman1999}
is a function of observed covariates $\boldsymbol{X}$ and individual-specific
reluctance to take treatment $V$, 
\begin{eqnarray*}
\theta_{\text{MTE}}(\boldsymbol{X},V) & = & E[Y_{1}-Y_{0}\vert\boldsymbol{X},V]\\
 & = & \underset{\text{constant term}}{\underbrace{\alpha_{1}-\alpha_{0}-0.5(\zeta_{1}-\zeta_{0})}}+(\boldsymbol{\beta_{1}}-\boldsymbol{\beta_{0}})^{\prime}\boldsymbol{X}+(\zeta_{1}-\zeta_{0})V.
\end{eqnarray*}
Within this MTE function, determining the sign of $\zeta_{1}-\zeta_{0}$,
if any, is typically a main focus of MTE estimation since it reveals
the direction of correlation between treatment effects and unobserved
reluctance to take treatment. 

To understand how this model of potential outcomes relates to the
standard linear IV model used in the majority of Mendelian randomization
applications, note that since the observed outcome is $Y=(1-A)Y_{0}+AY_{1}$,
we can write $E[Y\vert X,\boldsymbol{V}]=\delta(\boldsymbol{X})+\theta_{\text{MTE}}(\boldsymbol{X},V)\cdot A+\zeta_{0}V$
where the term $\zeta_{0}V$ is independent with the instrument $Z$
by Equation (2), and $\delta(\boldsymbol{X})$ is a function of covariates.
In the standard linear IV model under constant treatment effects,
the IV estimator would be consistent for $\alpha_{1}-\alpha_{0}$,
the average treatment effect over everyone in the population. However,
more generally this IV estimate will not be useful for inferring how
treatment effects vary with unobserved heterogeneity in treatment
choice. 

\subsection{Marginal Treatment Effects and Target Estimands}

One issue with directly estimating the MTE function $\theta_{\text{MTE}}(\boldsymbol{X},V)$
is that it depends on the unobserved variable $V$. However, an observable
quantity is the conditional mean of the observed outcome $m(\boldsymbol{x},p)=E_{{\cal P}}[Y\vert\boldsymbol{X}=\boldsymbol{x},\pi_{{\cal P}}(\boldsymbol{X},Z)=p]$,
which under Equations (1)--(3) can be written as 
\begin{equation}
m(\boldsymbol{x},p)=\boldsymbol{r}(\boldsymbol{x},p)^{\prime}\boldsymbol{\gamma}
\end{equation}
where $\boldsymbol{r}(\boldsymbol{x},p)$ is a known vector, and $\boldsymbol{\gamma}$
is an unknown vector such that 
\[
\boldsymbol{r}(\boldsymbol{x},p)=\begin{bmatrix}1\\
\boldsymbol{x}\\
p\\
\boldsymbol{x}p\\
p^{2}
\end{bmatrix},\,\,\text{and}\,\,\boldsymbol{\gamma}=\begin{bmatrix}\alpha_{0}\\
\boldsymbol{\beta_{0}}\\
\alpha_{1}-\alpha_{0}-0.5(\zeta_{1}-\zeta_{0})\\
\text{\ensuremath{\boldsymbol{\beta_{1}}-\boldsymbol{\beta_{0}}}}\\
0.5(\zeta_{1}-\zeta_{0})
\end{bmatrix}.
\]

\citet{Heckman1999} showed that the MTE function is identified as
the derivative 
\begin{eqnarray*}
\theta_{\text{MTE}}(\boldsymbol{x},v) & = & \partial m(\boldsymbol{x},p)\big/\partial p\vert_{p=v}\\
 & = & \underset{\text{constant term}}{\underbrace{\alpha_{1}-\alpha_{0}-0.5(\zeta_{1}-\zeta_{0})}}+(\boldsymbol{\beta_{1}}-\boldsymbol{\beta_{0}})^{\prime}\boldsymbol{x}+(\zeta_{1}-\zeta_{0})v
\end{eqnarray*}
and it represents the treatment effect for an individual who is indifferent
between being treated and untreated, at their given value of instrument
and covariates. This specification of the MTE means that the treatment
effect can vary linearly with covariates $\boldsymbol{X}$ and unobserved
heterogeneity $V$, and the intercept of the linear-in-$v$ MTE slope
can vary with covariates. 

By taking weighted averages of the vector of parameters $\boldsymbol{\gamma}$,
we can also calculate several scalar treatment effect estimands that
may be of interest; a detailed list of possible target estimands is
discussed in \citet{Heckman2005} and \citet{Mogstad2018}. In our
analysis, we focus on the following four target estimands. First,
the \textsl{average treatment effect} (ATE) effect measures the average
effect of treatment across the whole population, 
\begin{equation}
\theta_{\text{ATE}}(\boldsymbol{\gamma},{\cal P})=E_{{\cal P}}[Y_{1}-Y_{0}]=E_{{\cal P}}[\boldsymbol{r}_{\text{ATE}}(\boldsymbol{X})]^{\prime}\boldsymbol{\gamma}
\end{equation}
where $\boldsymbol{r}_{\text{ATE}}(\boldsymbol{X})=(0,\boldsymbol{0}^{\prime},1,\boldsymbol{X}^{\prime},1)^{\prime}$. 

The \textsl{average treatment effect on the treated} (ATT) measures
the average treatment effect for those who choose treatment, 
\begin{equation}
\theta_{\text{ATT}}(\boldsymbol{\gamma},{\cal P})=E_{{\cal P}}[Y_{1}-Y_{0}\vert A=1]={\cal P}(A=1)^{-1}E_{{\cal P}}[\boldsymbol{r}_{\text{ATT}}(\boldsymbol{X},Z,{\cal P})]^{\prime}\boldsymbol{\gamma}
\end{equation}
where $\boldsymbol{r}_{\text{ATT}}(\boldsymbol{X},Z,{\cal P})=(0,\boldsymbol{0}^{\prime},\pi_{{\cal P}}(\boldsymbol{X},Z),\pi_{{\cal P}}(\boldsymbol{X},Z)\boldsymbol{X}^{\prime},\pi_{{\cal P}}(\boldsymbol{X},Z)^{2})^{\prime}.$
Analogously, the \textsl{average treatment effect on the untreated}
(ATU) measures the average effect of treatment for those who choose
to be untreated, 
\begin{equation}
\theta_{ATU}(\boldsymbol{\gamma},{\cal P})=E_{{\cal P}}[Y_{1}-Y_{0}\vert A=0]={\cal P}(A=0)^{-1}E_{{\cal P}}[\boldsymbol{r}_{\text{ATU}}(\boldsymbol{X},Z,{\cal P})]^{\prime}\boldsymbol{\gamma}
\end{equation}
where $\boldsymbol{r}_{ATU}(\boldsymbol{X},Z,{\cal P})=(0,\boldsymbol{0}^{\prime},1-\pi_{{\cal P}}(\boldsymbol{X},Z),(1-\pi_{{\cal P}}(\boldsymbol{X},Z))\boldsymbol{X}^{\prime},1-\pi_{{\cal P}}(\boldsymbol{X},Z)^{2})^{\prime}$.

The \textsl{average selection on gains} (ASG) estimand simply finds
the difference between the ATT and the ATU, and summarizes how selection
into treatment is influenced by the treatment effect on average, 
\begin{eqnarray}
\theta_{\text{ASG}}(\boldsymbol{\gamma},{\cal P}) & = & E_{{\cal P}}[Y_{1}-Y_{0}\vert A=1]-E_{{\cal P}}[Y_{1}-Y_{0}\vert A=0]\nonumber \\
 & = & (E_{{\cal P}}[\boldsymbol{r}_{\text{ATT}}(\boldsymbol{X},Z,{\cal P})]{\cal P}(A=1)^{-1}-E_{{\cal P}}[\boldsymbol{r}_{\text{ATU}}(\boldsymbol{X},Z,{\cal P})]{\cal P}(A=0)^{-1})^{\prime}\boldsymbol{\gamma}.
\end{eqnarray}
In our application, if $\theta_{\text{ASG}}$ is negative, then this
would be consistent with health conscious behaviour where individuals
may act on private knowledge of their pre-existing health conditions,
and so those choosing not to consume alcohol excessively would be
more likely to be adversely affected in terms of higher blood pressure
if they did. 

It is easy to extend the potential outcomes model in Equation (3)
to allow for a more complex specification of the MTE function $\theta_{\text{MTE}}(\boldsymbol{X},V)$.
Although this would see the dimension of $\boldsymbol{\gamma}$ increase,
the target estimands in Equations (5)--(8) would still be weighted
averages of $\boldsymbol{\gamma}$. We derive our results under a
more general model that allows for higher-order polynomial terms and
interactions between observed covariates $\boldsymbol{X}$ and unobserved
heterogeneity $V$ (see Supplementary Material for further details).

\section{Efficient Estimation and Inference}

This section outlines two approaches to estimate the MTE parameters
$\boldsymbol{\gamma}$ and the target estimands in Section 2.2. First,
the ``conventional'' approach, and the one that is most widely used
in current practice, estimates the MTE parameters $\boldsymbol{\gamma}$
through a simple linear regression, and then calculates the sample
counterparts of Equations (5)--(8) to estimate the target estimands. 

Second, we propose an alternative ``efficient'' approach by calculating
efficient influence functions of $\boldsymbol{\gamma}$ and target
estimands to derive semiparametrically efficient estimates. Compared
with the conventional approach, our proposed efficient estimates of
target estimands will have a lower asymptotic variance and be robust
to sampling uncertainty in first stage propensity score estimation.
This robustness is particularly important when genetic instruments
only weakly influence treatment choice. 

\subsection{Estimation of Individual MTE Parameters}

The previous section shows that the estimands we are interested in
are weighted averages of the parameters in $\boldsymbol{\gamma}$.
Given the observed mean in Equation (4), $\boldsymbol{\gamma}$ is
equal to the coefficient of the \textsl{population-level} linear regression
of $Y$ on $\boldsymbol{r}(\boldsymbol{X},\pi_{{\cal P}}(\boldsymbol{X},Z))$.
Therefore, a feasible and intuitive method to estimate $\boldsymbol{\gamma}$
simply runs a linear regression of $Y$ on $\boldsymbol{r}(\boldsymbol{X},\pi_{\widehat{{\cal P}}}(\boldsymbol{X},Z))$
where $\pi_{\widehat{{\cal P}}}(\boldsymbol{X},Z)$ is an estimate
of the propensity score. We denote this estimate of $\boldsymbol{\gamma}$
as $\widehat{\boldsymbol{\gamma}}$, and refer to this as the ``conventional''
estimator because it is commonly used in applied work. For inference,
usual standard errors of $\widehat{\boldsymbol{\gamma}}$ from a typical
linear regression output will not be correct because they should be
corrected to account for additional variability from first stage propensity
score estimation. 

Is this conventional estimator $\widehat{\boldsymbol{\gamma}}$ the
best estimator of $\boldsymbol{\gamma}$? In terms of achieving a
low variance, it is efficient; its asymptotic variance is equal to
the semiparametric variance lower bound for the regression estimand
$\boldsymbol{\gamma}$. However, $\widehat{\boldsymbol{\gamma}}$
can be sensitive to misspecification of the propensity score $\pi_{{\cal P}}(\boldsymbol{X},Z)$,
or to the choice of tuning parameters if it is nonparametrically estimated
as in our case. More broadly, $\widehat{\boldsymbol{\gamma}}$ can
be sensitive to sampling uncertainty in first stage propensity score
estimates, which we demonstrate is more likely under weak instruments.
Given that genetic variants will typically explain only a low proportion
of variability in treatment choice, it is important in practice that
our estimates of MTE parameters are robust to weak instruments.

One way to achieve this is by constructing an estimator $\widetilde{\boldsymbol{\gamma}}$
of $\boldsymbol{\gamma}$ based on the efficient influence function
of the regression estimand $\boldsymbol{\gamma}$, which we call the
``efficient'' estimator. This estimator is semiparametrically efficient
for the regression estimand $\boldsymbol{\gamma}$. In our linear
model, $\widetilde{\boldsymbol{\gamma}}$ has a convenient closed
form expression
\[
\widetilde{\boldsymbol{\gamma}}=(\boldsymbol{\Omega}_{{\cal \widehat{P}}}+\boldsymbol{\Gamma}_{\widehat{{\cal P}}})^{-1}\boldsymbol{\Upsilon}_{\widehat{{\cal P}}}
\]
where $\boldsymbol{\Omega}_{{\cal \widehat{P}}}$, $\boldsymbol{\Gamma}_{\widehat{{\cal P}}}$,
and $\boldsymbol{\Upsilon}_{\widehat{{\cal P}}}$ are the sample counterparts,
with plug-in propensity score estimates, of the population means 
\begin{eqnarray*}
\boldsymbol{\Omega}_{{\cal P}} & = & E_{{\cal P}}[\boldsymbol{r}(\boldsymbol{X},\pi_{{\cal P}}(\boldsymbol{X},Z))\boldsymbol{r}(\boldsymbol{X},\pi_{{\cal P}}(\boldsymbol{X},Z))^{\prime}]\\
\boldsymbol{\Gamma}_{{\cal P}} & = & E_{{\cal P}}\big[\boldsymbol{R}(\boldsymbol{X},\pi_{{\cal P}}(\boldsymbol{X},Z))(A-\pi_{{\cal P}}(\boldsymbol{X},Z))\boldsymbol{r}(\boldsymbol{X},\pi_{{\cal P}}(\boldsymbol{X},Z))^{\prime}]\\
\boldsymbol{\Upsilon}_{{\cal P}} & = & E_{{\cal P}}[\boldsymbol{r}(\boldsymbol{X},\pi_{{\cal P}}(\boldsymbol{X},Z))Y]
\end{eqnarray*}
where $\boldsymbol{R}(\boldsymbol{X},\pi_{{\cal P}}(\boldsymbol{X},Z))=(1,\boldsymbol{X}^{\prime},2\pi_{{\cal P}}(\boldsymbol{X},Z))^{\prime}$. 
\begin{prop}[Efficient MTE parameter estimates]
 $\widetilde{\boldsymbol{\gamma}}$ is based on the efficient influence
function of the regression estimand $\boldsymbol{\gamma}$, and both
$\widehat{\boldsymbol{\gamma}}$ and $\widetilde{\boldsymbol{\gamma}}$
have the same asymptotic variance. 
\end{prop}
While both $\widehat{\boldsymbol{\gamma}}$ and $\widetilde{\boldsymbol{\gamma}}$
have the same asymptotic variance, $\widetilde{\boldsymbol{\gamma}}$
has the advantage of being robust to uncertainty in first stage propensity
score estimation. One property that feeds into the robustness of the
efficient estimator $\widetilde{\boldsymbol{\gamma}}$ is that it
is asymptotically equivalent to an estimator that would be obtained
if the true propensity score function $\pi_{{\cal P}}(\boldsymbol{X},Z)$
was known. In contrast the conventional estimator $\widehat{\boldsymbol{\gamma}}$
would not be efficient if the propensity score was known, even though
it is when the propensity score is estimated. 

For inference, it is also easy to calculate the standard error of
$\widetilde{\boldsymbol{\gamma}}$ because a consistent estimator
of the asymptotic variance of $\widetilde{\boldsymbol{\gamma}}$ is
the sample counterpart of the variance of the efficient influence
function divided by sample size. In particular, unlike for $\widehat{\boldsymbol{\gamma}}$,
there is no need to calculate the additional uncertainty related to
first stage propensity score estimation because this term is negligible
by design \citep[see, for example,][]{Hines2022}. 

\subsection{Estimation of Target Estimands}

For the four target estimands in Equations (5)--(8), we analogously
describe the ``conventional'' estimates as those that simply replace
the unknown $\boldsymbol{\gamma}$ with the conventional ordinary
least squares estimator $\widehat{\boldsymbol{\gamma}}$, the unknown
propensity scores $\pi_{{\cal P}}(\boldsymbol{X},Z)$ with parametric
or nonparametric estimates $\pi_{{\cal \widehat{P}}}(\boldsymbol{X},Z)$,
and population observed means with their sample counterparts. These
conventional estimates of the target estimands will\textsl{ not be
robust nor efficient. }
\noindent \begin{center}
{\footnotesize{}}%
\begin{tabular}{c|c|c}
\hline 
{\small{}Estimand} & {\small{}Conventional (plug-in) weights: $\widehat{\boldsymbol{\omega}}_{\text{J}}(.)$} & {\small{}Efficient weights: $\widetilde{\boldsymbol{\omega}}_{\text{J}}(.)$}\tabularnewline
\hline 
\hline 
{\small{}ATE} & {\footnotesize{}$E_{\widehat{{\cal P}}}[\boldsymbol{r}_{\text{ATE}}(\boldsymbol{X})]$} & {\footnotesize{}$\widehat{\boldsymbol{\omega}}_{\text{ATE}}-\widehat{\boldsymbol{\omega}}_{\text{ATE}}\boldsymbol{\Omega}_{\widehat{{\cal P}}}^{-1}\boldsymbol{\Gamma}_{\widehat{{\cal P}}}$}\tabularnewline
\hline 
{\small{}ATT} & {\footnotesize{}$\begin{array}{c}
\widehat{p}_{1}E_{\widehat{{\cal P}}}[\boldsymbol{r}_{\text{ATT}}(\boldsymbol{X},\pi_{\widehat{{\cal P}}}(\boldsymbol{X},Z))]\end{array}$} & {\footnotesize{}$\begin{array}{c}
\,\\
\widehat{\boldsymbol{\omega}}_{\text{ATT}}-\widehat{p}_{1}E_{\widehat{{\cal P}}}[\boldsymbol{r}_{\text{ATT}}(\boldsymbol{X},\pi_{{\cal \widehat{P}}}(\boldsymbol{X},Z))]^{\prime}\boldsymbol{\Omega}_{\widehat{{\cal P}}}^{-1}\boldsymbol{\Gamma}_{\widehat{{\cal P}}}+\\
\,\\
\widehat{p}_{1}E_{\widehat{{\cal P}}}\Big[\frac{\partial}{\partial\pi}\boldsymbol{r}_{\text{ATT}}(\boldsymbol{X},\pi_{{\cal \widehat{P}}}(\boldsymbol{X},Z))\cdot(A-\pi_{\widehat{{\cal P}}}(\boldsymbol{X},Z))\Big]\\
\,
\end{array}$}\tabularnewline
\hline 
{\small{}ATU} & {\footnotesize{}$\begin{array}{c}
\widehat{p}_{0}E_{\widehat{{\cal P}}}[\boldsymbol{r}_{\text{ATU}}(\boldsymbol{X},\pi_{\widehat{{\cal P}}}(\boldsymbol{X},Z))]\end{array}$} & {\footnotesize{}$\begin{array}{c}
\,\\
\widehat{\boldsymbol{\omega}}_{\text{ATU}}-\widehat{p}_{0}E_{\widehat{{\cal P}}}[\boldsymbol{r}_{\text{ATU}}(\boldsymbol{X},\pi_{{\cal \widehat{P}}}(\boldsymbol{X},Z))]^{\prime}\boldsymbol{\Omega}_{\widehat{{\cal P}}}^{-1}\boldsymbol{\Gamma}_{\widehat{{\cal P}}}+\\
\,\\
\widehat{p}_{0}E_{\widehat{{\cal P}}}\Big[\frac{\partial}{\partial\pi}\boldsymbol{r}_{\text{ATU}}(\boldsymbol{X},\pi_{{\cal \widehat{P}}}(\boldsymbol{X},Z))\cdot(A-\pi_{\widehat{{\cal P}}}(\boldsymbol{X},Z))\Big]\\
\,
\end{array}$}\tabularnewline
\hline 
{\small{}ASG} & {\footnotesize{}$\widehat{\boldsymbol{\omega}}_{\text{ATT}}-\widehat{\boldsymbol{\omega}}_{\text{ATU}}$} & {\footnotesize{}$\widetilde{\boldsymbol{\omega}}_{\text{ATT}}-\widetilde{\boldsymbol{\omega}}_{\text{ATU}}$}\tabularnewline
\hline 
\end{tabular}\\
~\\
{\small{}Table 1. The conventional estimator $\widehat{\boldsymbol{\omega}}_{\text{ }}^{\prime}(.)\widehat{\boldsymbol{\gamma}}$
and the efficient estimator $\widetilde{\boldsymbol{\omega}}^{\prime}(\cdot)\widehat{\boldsymbol{\gamma}}$
for each treatment effect estimand, where $\widehat{p}_{0}=\widehat{{\cal P}}(A=0)^{-1}$
and $\widehat{p}_{1}=\widehat{{\cal P}}(A=1)^{-1}$.}{\small\par}
\par\end{center}

In Table 1, we derive efficient estimates of target estimands. Both
the conventional and efficient estimates of the target estimands are
weighted averages of the ordinary least squares estimator $\widehat{\boldsymbol{\gamma}}$;
Table 1 shows the weights for the efficient estimates involve only
a couple of more terms compared to the conventional estimates, and
therefore the efficient estimates are also easy to compute. 
\begin{prop}[Efficient target parameter estimates]
 For target estimands in Equations (5)--(8), which are defined as
functions of the regression estimand $\boldsymbol{\gamma}$, the semiparametrically
efficient estimates are given in Table 1. 
\end{prop}
Using efficient estimators of target estimands ensures low variance
in estimation, and they also have the robustness property discussed
above where the first order asymptotic distribution of the efficient
estimators are equivalent to the oracle case where the true propensity
score is known. The benefit of robustness to first stage nuisance
estimation has been extensively explored, especially in the literature
on ``doubly robust'' estimation of treatment effects \citep{Scharfstein1999},
where asymptotic normality holds even under slow convergence rates
from machine learning estimators of first stage nuisance functions
\citep{Chernozhukov2018a}.

For our applied focus, this robustness property is useful when estimating
MTEs using a genetic risk score as an instrument. The propensity score
can be precisely measured where the distribution of genetic risk scores
are heavily concentrated, but estimating the tails of the propensity
scores are based on many fewer observations (see Figure 4 below).
As a result, sampling uncertainty in propensity score estimation will
be much greater in the tails, providing a challenge to reliably estimate
the MTE curve over the full range of observed propensity scores. Importantly,
efficient estimators are less sensitive to this type of weak instrument
setting. 

\section{Simulation Study}

We performed a simulation study to illustrate how efficient methods
can provide a more robust analysis under varying levels of instrument
strength. We start by describing the simulation design, and then present
results showing how the estimation performance of MTE parameters and
target estimands varies with instrument strength under a nonparametrically
estimated propensity score. 

\subsection{Design}

We generated a binary covariate $X\sim\text{Binomial}(1,0.5)$, and
an instrument $Z\sim N(0,1)$ as exogenous variables. Let $(W_{0},W_{1})$
be bivariate normal random variables with variance $0.2$ and correlation
$0.2$, and let $V$ be a standard uniform random variable. We generated
the potential outcomes as 
\begin{eqnarray*}
Y_{0} & = & 0.3+0.1X-0.3(V-1/2)+\sqrt{0.2}W_{0}\\
Y_{1} & = & 0.5+0.2X+0.3(V-1/2)+\sqrt{0.2}W_{1}
\end{eqnarray*}
so that Equation (3) was satisfied with $\alpha_{0}=0.3$, $\alpha_{1}=0.5$,
$\beta_{0}=0.1$, $\beta_{1}=0.2$, $\zeta_{0}=-0.3$, and $\zeta_{1}=0.3$. 

The treatment indicator was generated as $A=1$ only if $\Phi(\boldsymbol{q}(X,Z)^{\prime}\boldsymbol{\eta}))>V$,
where $\Phi$ is the cumulative distribution function of a standard
normal random variable. We set $\boldsymbol{q}(X,Z)=(1,X,Z,XZ)^{\prime}$
and the true parameter values as $\boldsymbol{\eta}=(0,-0.2,\bar{\eta},-0.2\bar{\eta})^{\prime}$
for varying $\bar{\eta}>0$. The further $\bar{\eta}$ is away from
zero, the greater the instrument strength and the greater the variation
in the propensity score for a given covariate level. 

Given this simulation design, the correct specification of the observable
mean function $m(x,p)$ in Equation (4) was
\[
m(x,p)=0.3+0.1x-0.1p+0.1xp+0.3p^{2}.
\]
Differentiating $m(x,p)$ with respect to $p$ and evaluating at $p=v$
gives the true MTE function 
\[
\theta_{\text{MTE}}(x,v)=-0.1+0.1x+0.6v.
\]
From this, we see that the MTE curve is increasing in unobserved heterogeneity
$V$. If \textsl{lower} values of the outcome variable are favourable,
then this describes a setting of ``selection on gains'', i.e.\@
those who benefit most from treatment are those who choose to be treated.
By averaging the MTE function over covariate $X$ and unobserved heterogeneity
$V$, we find that the true ATE was 0.25, and the true values of other
target estimands varied with instrument strength $\bar{\eta}$. As
we increase instrument strength from $\bar{\eta}=0.2$ to $\bar{\eta}=1$,
the ATT ranged from 0.088 to 0.135, the ATU from 0.388 to 0.352, and
the ASG from -0.300 to -0.217. 

\subsection{Estimation of MTE Parameters}

The propensity score was estimated by a nonparametric kernel regression
using the \texttt{np} R package \citep{Hayfield2008}. One advantage
of efficient estimation is that it is designed to be insensitive as
the propensity score estimates $\pi_{\widehat{{\cal P}}}(\boldsymbol{X},Z)$
move slightly away from the truth $\pi_{{\cal P}}(\boldsymbol{X},Z)$.
What this means in practice is that the efficient estimator of MTE
parameters $\widetilde{\boldsymbol{\gamma}}$ should be more accurate
than the conventional estimator $\widehat{\boldsymbol{\gamma}}$ when
instruments are relatively weak, where some regions of the propensity
score are estimated with fewer observations. This advantage of the
efficient approach over the conventional approach is what we demonstrate
in this section.

For bandwidth selection, for each simulation run we computed bandwidth
constants using cross-validation on $3$ random samples of size $1,000$
from a full sample of $n=10,000$, and the bandwidths for the full
sample size were scaled down from the constant according to the optimal
rate; for further details, see the discussion of the \texttt{bwscaling}
argument in \citet{Hayfield2008}. The robustness property of the
efficient approach means that its estimates should be relatively insensitive
to the choice of bandwidth. Supplementary Figures S4 and S5 show that
efficient estimators of MTE and target parameters continue to have
low bias when the chosen bandwidth is away from the optimal level,
especially when oversmoothing. 
\noindent \begin{center}
\includegraphics[height=5cm]{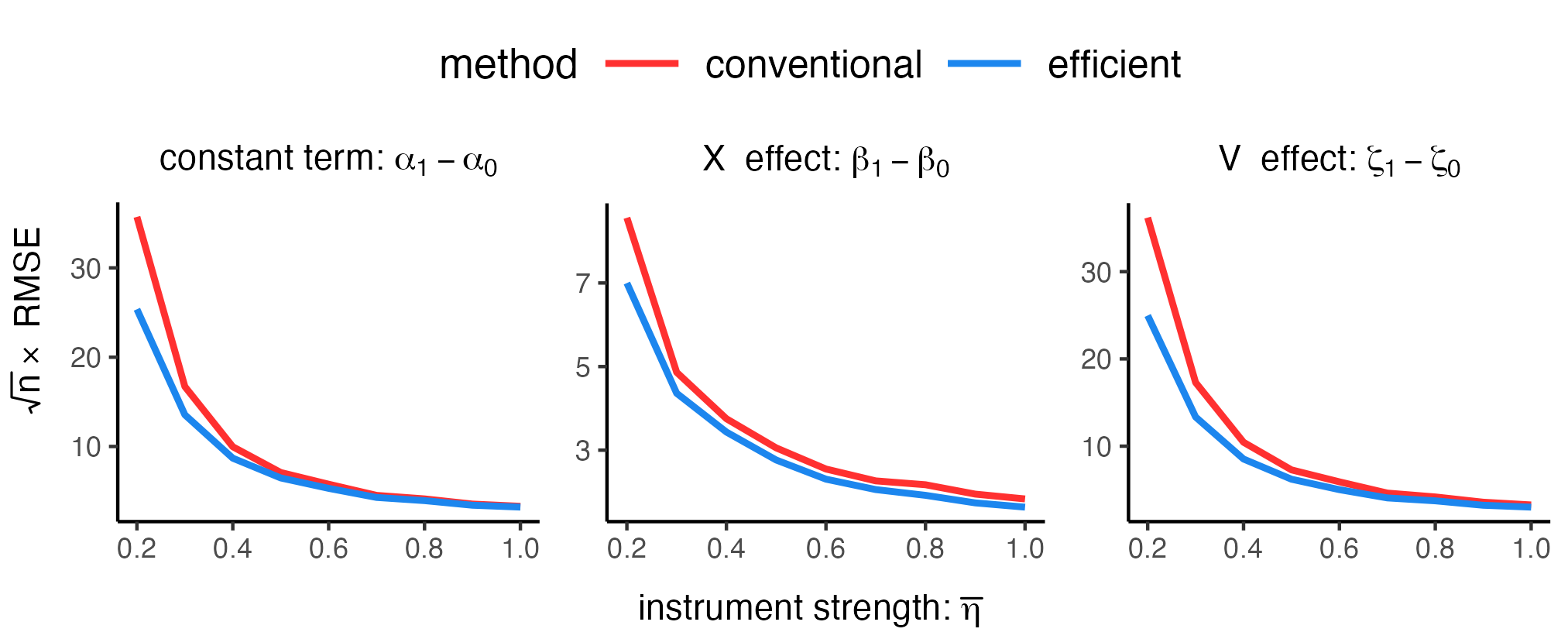}~\\
{\small{}Figure 1. Root mean squared error (RMSE) of MTE parameter
estimates over 1000 simulation experiments for each value of instrument
strength $\bar{\eta}$. }{\small\par}
\par\end{center}

Figure 1 shows that the efficient estimates of MTE parameters are
better than the conventional estimates in terms of achieving a lower
mean squared error (MSE). The improvement from using the efficient
approach is greater when instruments are weaker. As instrument strength
increases, the MSE from both methods converge toward zero as expected.
Since the efficient and conventional approach have the same asymptotic
variance, we might expect any differences in MSE to be driven by bias.
Supplementary Figure S1 shows that the bias is indeed lower for efficient
estimates when instruments are sufficiently strong $(\bar{\eta}\geq0.4)$,
and especially for the $V$ effect, the conventional estimates were
considerably more biased even in the strongest instrument case $(\bar{\eta}=1)$. 

\subsection{Estimation of Target Estimands}

Given that the target estimands ATE, ATT, ATU, and ASG are linear
functions of the MTE parameters $\boldsymbol{\gamma}$, we would expect
the estimation performance of MTE parameters is closely linked with
the estimation performance of our target estimands. However, the asymptotic
variances of the conventional and efficient approaches are now no
longer the same, and hence we now expect a greater advantage from
efficient estimation in terms of MSE performance. 

Figure 2 shows this is the case for all four estimands, ATT, ATT,
ATU, and ASG. For example, under very weak instruments ($\bar{\eta}=0.2$),
the average root-MSE (RMSE) of the efficient estimates was around
40\% lower than that of the conventional ones for all target estimands.
Supplementary Figure S2 shows the bias of efficient estimates hovered
closely around zero for all estimands when instruments were sufficiently
strong ($\bar{\eta}\geq0.6$). In certain weak instrument cases, the
bias of the conventional approach was lower than that of the efficient
approach, especially for the ATT estimand. Thus, given its superior
MSE performance in Figure 2, the variance of the efficient estimates
was considerably lower under weak instruments. 
\noindent \begin{center}
\includegraphics[height=5cm]{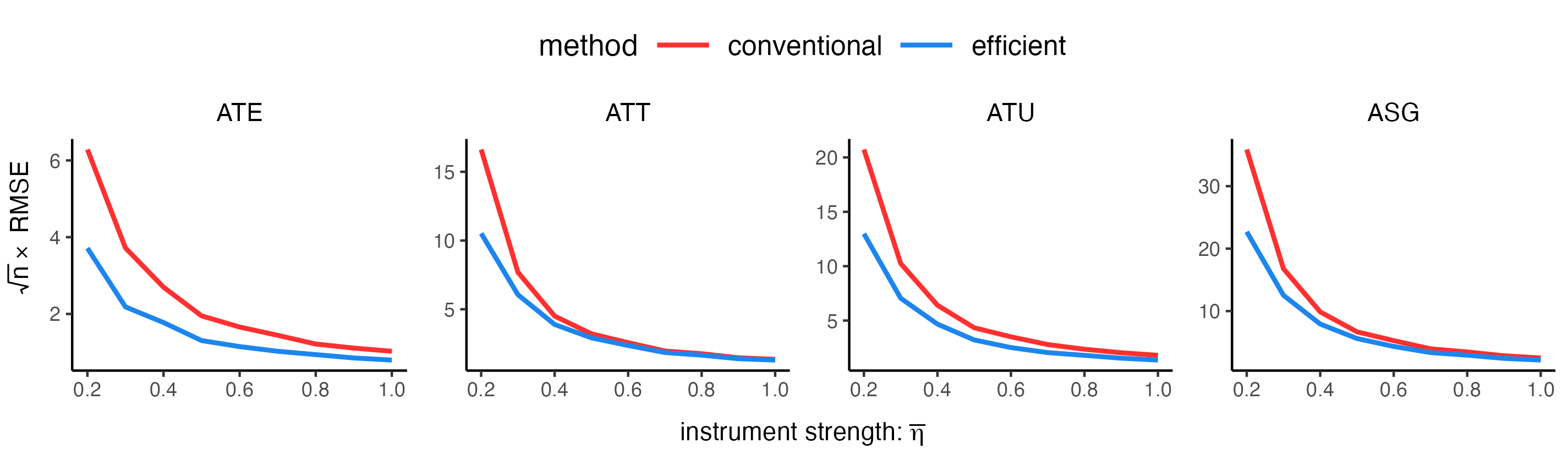}~\\
{\small{}Figure 2. Root mean squared error (RMSE) of target estimand
point estimates over 1000 simulation experiments for each value of
instrument strength $\bar{\eta}$.}{\small\par}
\par\end{center}

To compare estimation performance on the MTE estimand itself, Figure
3 plots the average RMSE of estimates of the MTE estimand averaged
over covariates $X$, i.e.\@ $E_{{\cal P}}[Y_{1}-Y_{0}\vert\,V]$. 

Figure 3 shows that the MTE estimates had a lower MSE than the conventional
estimates. Both estimators are more precise where most observations
are concentrated (where $V$ is around 0.5), and are less precise
when estimating the treatment effect for those with extreme values
of $V$. For the case of $\bar{\eta}=0.2$, observed propensity scores,
on average, ranged from 0.211 to 0.769, and under this greater sampling
uncertainty in the tails of propensity scores, the gap between the
MSE of the two approaches is greater. 
\noindent \begin{center}
\includegraphics[height=5cm]{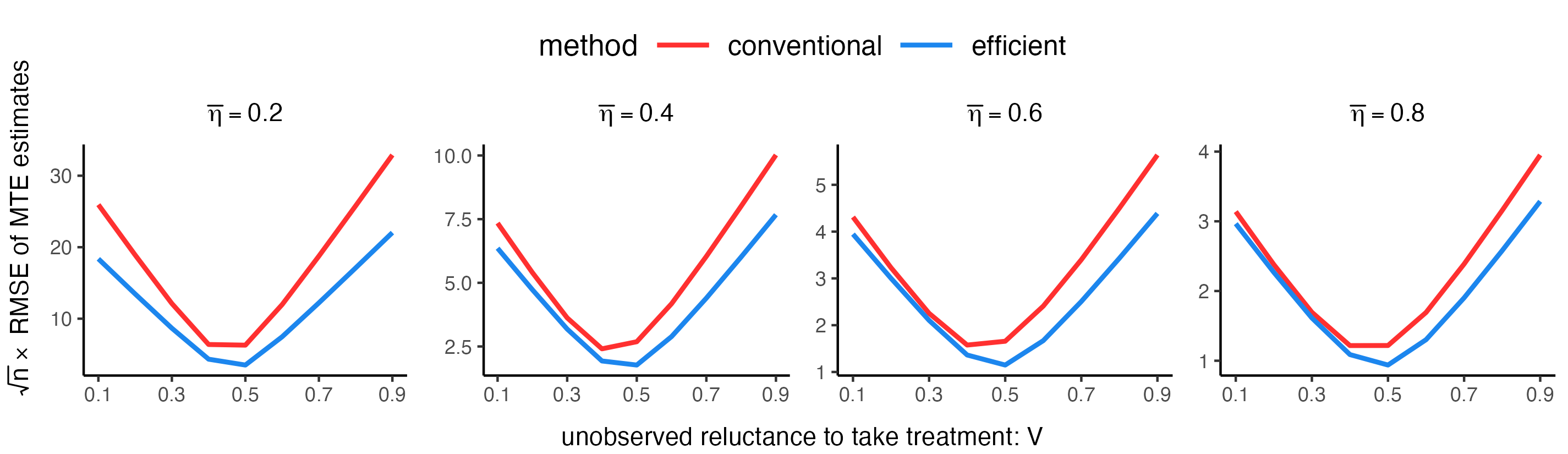}~\\
{\small{}Figure 3. Root mean squared error (RMSE) of MTE estimates
varying with $V$ and averaged over covariates, $E[Y_{1}-Y_{0}\vert V]$. }{\small\par}
\par\end{center}

The observed propensity scores had close to full support over the
$0$ to $1$ range when instruments were sufficiently strong ($\bar{\eta}\geq0.8$);
when $\bar{\eta}=0.8$, 1.6\% of of samples had a propensity score
less than 0.05 and 1.1\% of samples had a propensity score over 0.95.
In this setting, not only is the efficient estimator better in terms
of MSE, but it performs better in terms of bias (Supplementary Figure
S3). 

Overall, our simulation results highlight the benefits of efficient
estimation of MTE parameters and related target estimands. In particular,
efficient estimates have a lower MSE in finite samples compared with
the conventional plug-in propensity score approach, and in most cases
a lower bias. The efficient estimates are also more robust to weak
instrument settings, and to the choice of tuning parameters in nonparametric
propensity score estimation. 

\section{Empirical Application: Estimating Effects of Excessive\protect \\
Alcohol Consumption on Blood Pressure}

Estimating the causal effects of excessive alcohol consumption is
complicated by self-selection into drinking and heterogeneity in responses
to health advice. A generalized Roy model can capture the dynamic
that alcohol consumption behaviour may influence, and be influenced
by, worse health outcomes. The UK National Health Service advises
adults to drink no more than 14 units of alcohol per week (where 1
unit is 8g of pure alcohol). Therefore, we work with a binary treatment
indicator for ``excessive alcohol consumption'' defined as self-reported
alcohol consumption of at least 14 units of alcohol a week. Systolic
blood pressure measurements were used as the outcome variable. Here,
unobserved heterogeneity $V$ can be thought as describing heterogeneity
in ``health consciousness'' levels: within groups sharing similar
observable covariate values or characteristics, those with high values
of $V$ are less likely to excessively consume alcohol compared with
those with low values of $V$. 

We use a genetic risk score for alcohol consumption as a continuous
instrument for excessive alcohol consumption, and estimate MTE parameters
and other target estimands to study the potentially heterogeneous
effects of excessive alcohol consumption on systolic blood pressure.
In Supplementary Material, we present results under a more general
model for the MTE that allows the effects of unobserved heterogeneity
$V$ to differ for women and men, however we did not find evidence
for a sex difference in the MTE slope-in-$V$.

\subsection{Data}

\subsubsection{UK Biobank and Alcohol Consumption}

We considered data on 367,703 genotyped participants in the UK Biobank
of European ancestry. After removing samples that provided incomplete
data on alcohol consumption, systolic blood pressure, and antihypertensive
medication use, we were left with complete data on 294,563 individuals.
The treatment indicator of excessive alcohol consumption was defined
as a binary variable that equalled 1 when an individual's self-reported
alcohol consumption was more than 14 units of alcohol per week on
average, and zero otherwise. Under this definition, 53.5\% of the
294,563 participants were consuming alcohol excessively. 

\subsubsection{Instrument Selection}

The motivation of using genetic instruments is that genetic variants,
which are fixed at conception, may be less prone to unmeasured confounding
and reverse causation, and may provide a source of exogenous variation
in treatment choice \citep{Smith2003}. Our analysis used a genetic
risk score for alcohol consumption composed of 93 genetic variants,
as an instrumental variable for excessive alcohol consumption, as
previously used in \citet{Topiwala2022}. The 93 genetic variants
were each marginally associated with alcohol consumption (log drinks
per week) at a p-value less than the genome-wide significance threshold
$5\times10^{-8}$ in a large published genome-wide association study
\citep{Liu2019}. The individual variants used in the genetic risk
score were not strongly correlated (the correlation between any two
variants was $R^{2}\leq0.1$). The genetic risk score was robustly
associated with a continuous measure of alcohol consumption in UK
Biobank participants on average, with an F-statistic of 1862. 

\subsubsection{Outcome Selection}

We considered systolic blood pressure (SBP) as the outcome variable.
For individuals on antihypertensive medication, we added 10mmHg to
their SBP measurements since most antihypertensive drugs lower systolic
blood pressure by around that amount \citep{Dimmitt2019}. The average
SBP over the 294,563 UK Biobank participants considered for our analysis
was 139.98 mmHg. Approximately 45.5\% of the sample had SBP over 140
mmHg, a commonly used indicator for hypertension. There were sex differences
in SBP measurements; the average SBP estimate was 136.60 mmHg among
female participants, and 143.51 mmHg among male participants. We standardized
the outcome so that each unit change in the outcome represents a standard
deviation change in SBP, where 1 standard deviation of SBP equalled
19.759 mmHg. 

\subsubsection{Covariate Selection}

We considered sex as a single covariate $\boldsymbol{X}=X_{woman}$,
a binary indicator for female participants. We also considered age
and geographical location (a binary indicator if they lived north
of Stoke in the UK) as potential covariates in addition to sex, but
their impact on alcohol consumption behaviour was not estimated to
be relevant. Sex would seem to be a relevant covariate since women
and men metabolize alcohol differently due to biological differences
\citep{Graham1998,Briasoulis2012}. Around 51.1\% of participants
were women, and around 48.9\% were men. 

\subsection{Predictors of Excessive Alcohol Consumption}

Estimation of MTEs is more precise when there is large variation in
the propensity score caused by the genetic instrument and covariates.
Figure 4 plots how the propensity to drink excessively varies with
sex and the genetic risk score. 
\noindent \begin{center}
\includegraphics[width=17.5cm]{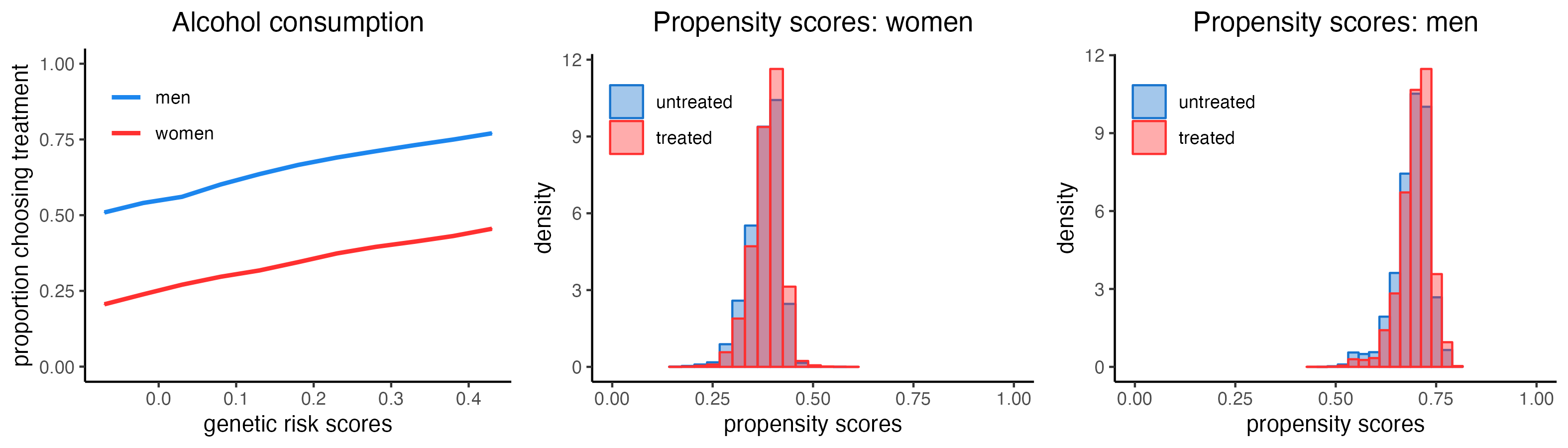}~\\
{\small{}Figure 4. }{\footnotesize{}Excessive alcohol consumption
(treatment) varying with the genetic risk score and sex (left), and
the common support of observed propensity scores for women (middle)
and for men (right). }{\footnotesize\par}
\par\end{center}

Figure 4 shows that for both men and women, a higher genetic risk
score is associated with a greater propensity for excessive alcohol
consumption, which is unsurprising given that the genetic risk score
was calculated for alcohol consumption. For a given genetic risk score,
women were considerably less likely to drink than men. We conclude
that sex and the weighted genetic risk score are conditionally relevant
predictors of excessive alcohol consumption. While there is an even
match in the distributions of treated and untreated groups, the variation
in the propensity score conditional on sex is quite limited. Therefore
estimating a more general MTE curve that allows for $V\times X$ interactions
could be challenging (see Supplementary Material for results under
this more general model).

\subsection{Results}

The propensity score was estimated using a nonparametric kernel regression
(of excessive alcohol consumption on sex and the genetic risk score)
using the \texttt{np} R package \citep{Hayfield2008}. For bandwidth
selection, we computed bandwidth constants using cross-validation
on $20$ random samples (without replacement) of size $10,000$ from
a full sample of $n=294,563$, and the bandwidths for the full sample
size were scaled down from the constant according to the optimal rate;
for further details, see the discussion of the \texttt{bwscaling}
argument in \citet{Hayfield2008}. 
\noindent \begin{center}
\begin{tabular}{c|c|c|c}
\hline 
{\small{}MTE Parameter} & {\small{}Constant term} & \multicolumn{1}{c|}{{\small{}Woman $(X)$}} & \multicolumn{1}{c}{{\small{}Health consciousness $(V)$}}\tabularnewline
\hline 
{\small{}Estimate} & {\small{}1.905} & {\small{}-0.815} & {\small{}-2.215}\tabularnewline
{\small{}Standard Error} & {\small{}0.728} & {\small{}0.340} & {\small{}1.080}\tabularnewline
{\small{}95\% CI-Lower} & {\small{}0.477} & {\small{}-1.482} & {\small{}-4.331}\tabularnewline
{\small{}95\% CI-Upper} & {\small{}3.332} & {\small{}-0.148} & {\small{}-0.098}\tabularnewline
\hline 
\end{tabular}\\
{\small{}~}\\
{\small{}Table 2. Efficient estimates of MTE parameters from the model
$\theta_{\text{MTE}}(X,V)\sim\text{constant term}+X_{\text{woman}}+V_{\text{health-consciousness}}$.}{\small\par}
\par\end{center}

From Table 2, the MTE of excessive alcohol consumption on systolic
blood pressure is estimated to be positive for most individuals, apart
from those with high values of health consciousness $V$. The MTE
curve is estimated to be decreasing in $V$, so that any adverse effect
of excessive alcohol consumption on systolic blood pressure is lower
for health conscious individuals. This suggests evidence of ``reverse
selection on gains'': individuals that are less health conscious
may experience worse health outcomes (greater changes in systolic
blood pressure) due to excessive alcohol consumption than more health
conscious individuals. 

We also estimate sex differences in the treatment effect, where the
adverse effect of excessive alcohol consumption on systolic blood
pressure is estimated to be worse for men than for women. One potential
explanation of this finding is that excessive alcohol drinking may
involve episodic drinking (i.e.\@ binge drinking) in a higher proportion
of men and those with low health consciousness, and such drinking
has a greater effect on blood pressure compared to steady levels of
high alcohol consumption.

Table S1 in Supplementary Material presents MTE parameter estimates
from a more general model allowing for interaction effects between
sex $(X)$ and unobserved health consciousness $(V)$. The effects
of sex and health consciousness are directionally the same as those
estimated in Table 2, although there is larger uncertainty in the
estimates. We estimate no interaction effect between $X$ and $V$,
so we are unable to conclude the slope of the MTE curve in terms of
health consciousness levels $V$ is different for female participants
compared to male participants. 
\noindent \begin{center}
\begin{tabular}{c|c|c|c|c}
\hline 
{\small{}Estimand} & {\small{}ATE} & {\small{}ATT} & \multicolumn{1}{c|}{{\small{}ATU}} & \multicolumn{1}{c}{{\small{}ASG}}\tabularnewline
\hline 
{\small{}Estimate} & {\small{}0.379} & {\small{}0.952} & {\small{}-0.282} & {\small{}1.235}\tabularnewline
{\small{}Standard Error} & {\small{}0.051} & {\small{}0.341} & {\small{}0.349} & {\small{}0.684}\tabularnewline
{\small{}95\% CI-Lower} & {\small{}0.279} & {\small{}0.283} & {\small{}-0.966} & {\small{}-0.105}\tabularnewline
{\small{}95\% CI-Upper} & {\small{}0.479} & {\small{}1.621} & {\small{}0.401} & {\small{}2.575}\tabularnewline
\hline 
\end{tabular}\\
{\small{}~}\\
{\small{}Table 3. Efficient estimates of target estimands from the
model $\theta_{\text{MTE}}(X,V)\sim\text{constant term}+X_{\text{woman}}+V_{\text{health-consciousness}}$.}{\small\par}
\par\end{center}

Table 3 presents estimates of target estimands. For the ATE, genetically-predicted
excessive alcohol consumption is associated with higher systolic blood
pressure measurements by 0.379 standard deviations (p-value < 0.001),
which is 7.489 mmHg. This is consistent with experimental evidence
that excessive alcohol consumption leads to higher blood pressure
levels \citep{Tasnim2020}. We also note that the observational association
between excessive alcohol consumption and systolic blood pressure
is $E_{\widehat{{\cal P}}}[Y\vert A=1]-E_{\widehat{{\cal P}}}[Y\vert A=0]=0.186$,
or 3.675 mmHg.

The ATT was estimated to be positive (p-value: 0.005), and while there
is greater uncertainty regarding the ATU estimate, results broadly
suggest a positive estimate of the ASG (p-value: 0.071), which is
consistent with the reverse selection on gains results inferred directly
from MTE parameters in Table 2. Table S2 in Supplementary Material
shows that in an MTE model allowing for $X\times V$ interaction effects,
the estimates of the target estimands ATE, ATT, ATU, and ASG were
similar to those estimated in Table 2. This supports the results of
Table S1 that suggested no evidence for sex differences in the MTE
slope-in-$V$. 
\noindent \begin{center}
\includegraphics[height=5.5cm]{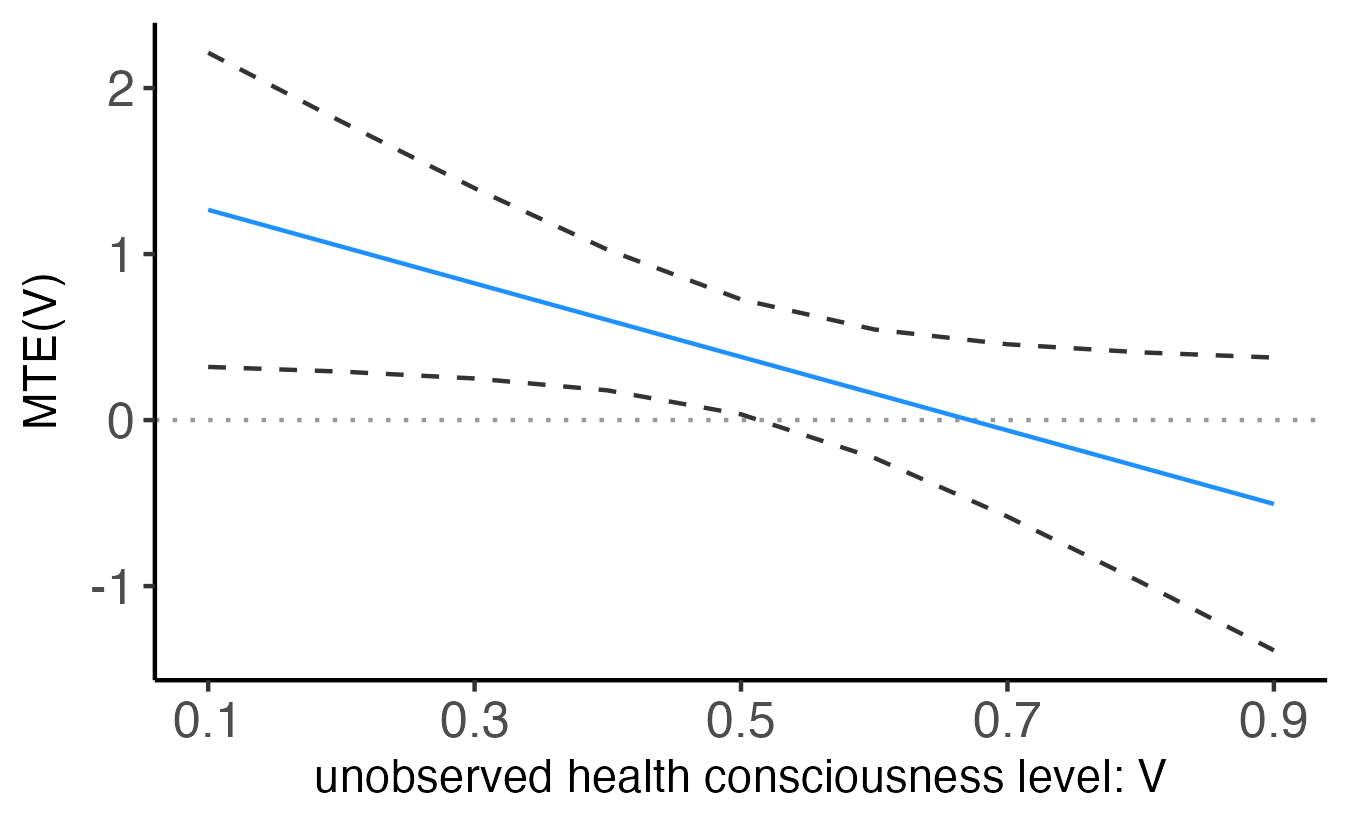}~\\
{\small{}Figure 5. Efficient MTE estimates varying with $V$ and averaged
over covariates, $E[Y_{1}-Y_{0}\vert V]$ (solid blue line). The dashed
black curves indicate 95\% asymptotic confidence intervals. }{\small\par}
\par\end{center}

Figure 5 illustrates that the harmful effects of excessive alcohol
consumption are estimated to be decreasing in health consciousness
levels. The 95\% asymptotic confidence intervals (dashed black lines)
show how the MTE is estimated with less precision at the tail ends
of $V$ which again reflects the lack of support in the propensity
score at those regions. Supplementary Figure S6 plots the MTE curve
under a more general model allowing for $X\times V$ interaction effects,
and this estimated curve is very close to that of Figure 5 which supports
our MTE specification with separable $X$ and $V$ effects. 

In summary, under a generalized Roy model that uses a genetic risk
score as an instrument for alcohol consumption, our results suggest
evidence of reverse selection on gains in alcohol consumption: less
health conscious individuals that are prone to consume alcohol excessively
are those that may experience worse health outcomes (in terms of higher
systolic blood pressure) as a result. Additionally, sex was estimated
to be a relevant covariate for predicting the harmful effect of excessive
alcohol consumption, with the harmful effect being worse for men than
for women. 

However, we urge caution in the interpretation of our estimates. For
our treatment effect estimates we compared the difference between
genetically defined subgroups that may reflect a more long-term, sustained
effect of high alcohol consumption on blood pressure \citep{Burgess2023}.
As such, these estimates are not easily comparable to those from a
short-term intervention of a clinical trial \citep{Tasnim2020}.

\section{Conclusion}

The MTE framework \citep{Heckman1999,Heckman2005} is able to model
the plausible scenario that the effect of a treatment may vary with
both observed and unobserved heterogeneity that influences selection
into treatment. It is a popular method to estimate causal effects
in observational studies, and has been applied to study a wide range
of topics including education and healthcare; see \citet{Shea2023}
for a detailed list of applications. 

As with any instrumental variable method, the credibility of an applied
analysis relies on the instrument being valid and relevant. In epidemiology,
the Mendelian randomization paradigm has popularized the use of genetic
variants as instrumental variables. However, the vast majority of
Mendelian randomization applications have estimated causal effects
in a \emph{homogenous} effects model where the treatment effect is
assumed to be the same for all individuals, due in part to easy-to-use
software that can be applied to readily available summary data from
genome-wide association studies. 

The contribution of our work was three-fold. First, we demonstrated
the use of genetic instruments to estimate \emph{heterogeneous} causal
effects in the MTE framework, which studies the relationship between
treatment choice and treatment effects. This introduces a credible
approach to study effect heterogeneity in Mendelian randomization
studies with a binary choice exposure. 

Second, genetic instruments will typically be weak in the sense that
they are unlikely to cause large shifts in the propensity scores of
treatment, which can lead to issues for estimating the MTE function.
Thus, to reduce sensitivity to the first stage uncertainty at the
tails of propensity score estimates, we derived the efficient influence
functions of target estimands in linear-in-parameter models of MTEs.
Our work has clarified that efficient estimates can be more robust
for estimation of target estimands, while being computationally straightforward.
We have also developed an R package (\href{http://github.com/ash-res/efficient.mte/}{github.com/ash-res/efficient.mte/})
for straightforward application of the methods we discussed. While
this work was motivated by the specific case of a genetic instrument,
the methods proposed in this paper are applicable to a wider class
of instrumental variable studies. 

Third, our empirical analysis suggests genetic evidence for heterogeneity
in the harmful effects of excessive alcohol consumption, such that
individuals that are prone to drink alcohol excessively may experience
greater changes in systolic blood pressure compared to individuals
that are less prone to drink excessively. Given that high alcohol
consumption is known to be significantly related with many adverse
health conditions globally \citep{Anderson2023}, this highlights
the importance of policies to target those more vulnerable to excessively
consume alcohol. 

\subsection*{Funding}

SB was supported by the Wellcome Trust (225790/Z/22/Z). This research
was funded by the United Kingdom Research and Innovation Medical Research
Council (MCUU- 00002/7).

{\small{}\bibliographystyle{chicago}
\bibliography{heterogeneity}
}{\small\par}

\newpage{}

\setcounter{page}{1}
\setcounter{equation}{0} 
\noindent \begin{center}
{\Huge{}Supplementary Material}{\Huge\par}
\par\end{center}

\noindent \begin{center}
{\large{}Ashish Patel, Francis J.\@ DiTraglia \& Stephen Burgess}{\large\par}
\par\end{center}

\subsection*{Notation and Preliminaries}

{\small{}For any vectors of observed variables $A$ and $B$, let
$var_{{\cal P}}(A)=E_{{\cal P}}[(A-E_{{\cal P}}(A))(A-E_{{\cal P}}(A))^{\prime}]$,
and let $cov_{{\cal P}}(A,B)=E_{{\cal P}}[(A-E_{{\cal P}}(A))(B-E_{{\cal P}}(B))^{\prime}]$.
We also note their sample averages as $var_{\widehat{{\cal P}}}(A)$
and $cov_{\widehat{{\cal P}}}(A,B)$. Then, given the linear model
$m(x,p)=r(x,p)^{\prime}\gamma$ for the observable mean $m(x,p)=E_{{\cal P}}[Y\vert X=x,\pi_{{\cal P}}(X,Z)=p]$,
we can define a target estimand $\gamma=\gamma({\cal P})$ as 
\[
\gamma({\cal P})=var_{{\cal P}}(r(X,\pi_{{\cal P}}(X,Z)))^{-1}cov_{{\cal P}}(r(X,\pi_{{\cal P}}(X,Z)),Y).
\]
To calculate the efficient influence function of $\gamma({\cal P})$,
we follow the }\textsl{\small{}point mass contamination}{\small{}
strategy recommended by \citet{Ichimura2015} and \citet{Hines2022}.
Let $\widetilde{{\cal P}}$ denotes a fixed distribution that has
a point mass at a single observation $\widetilde{o}=(\tilde{a},\tilde{x},\tilde{y},\tilde{z})$.
We wish to perturb the estimand $\gamma({\cal P})$ in the direction
$\widetilde{{\cal P}}$ by using a mixture model ${\cal P}_{t}=t\widetilde{{\cal P}}+(1-t){\cal P}$
indexed by $t\in[0,1]$. Then, the efficient influence function at
observation $\widetilde{o}$ is given by the Gateaux derivative 
\[
\phi(\widetilde{o},{\cal P})=\frac{\partial\gamma({\cal P}_{t})}{\partial t}\Big\vert_{t=0}.
\]
}{\small\par}

{\small{}The efficient influence function $\phi(\widetilde{o},{\cal P})$
characterises how sensitive the estimand $\gamma({\cal P})$ is to
changes in the data-generating distribution ${\cal P}$ at observation
$\widetilde{o}$, and it plays an important part in the first-order
asymptotic properties of an estimator of $\gamma({\cal P})$. The
term
\[
\frac{1}{\sqrt{n}}\sum_{i=1}^{n}\phi(O_{i},\widehat{{\cal P}})
\]
is called the }\textsl{\small{}drift term}{\small{} of a von Mises
expansion of the nonparametric estimator of $\gamma({\cal P})$. One
way to construct an efficient estimator of $\gamma({\cal P})$ is
to choose the estimator $\gamma(\widehat{{\cal P}})$ that sets this
drift term to zero. This is the approach we now take. }{\small\par}

{\small{}Setting the drift term to zero, we find the estimator $\widetilde{\gamma}\equiv\gamma(\widehat{{\cal P}})$
as the solution to the estimating equation
\begin{eqnarray*}
0 & = & -var_{\widehat{{\cal P}}}(r(X,\pi_{\widehat{{\cal P}}}(X,Z)))^{-1}cov_{\widehat{{\cal P}}}\Big(\frac{\partial r(X,\pi_{{\cal \widehat{P}}}(X,Z))}{\partial\pi}(A-\pi_{\widehat{{\cal P}}}(X,Z)),r(X,\pi_{\widehat{{\cal P}}}(X,Z))\Big)\gamma(\widehat{{\cal P}})\\
 &  & -\gamma(\widehat{{\cal P}})+var_{\widehat{{\cal P}}}(r(X,\pi_{\widehat{{\cal P}}}(X,Z)))^{-1}cov_{\widehat{{\cal P}}}(r(X,\pi_{\widehat{{\cal P}}}(X,Z)),Y).
\end{eqnarray*}
Note that this is a different estimator to the simple conventional
(OLS) estimator 
\[
\widehat{\gamma}=var_{\widehat{{\cal P}}}(r(\pi_{\widehat{{\cal P}}}(X,Z),X))^{-1}cov_{\widehat{{\cal P}}}(r(\pi_{\widehat{{\cal P}}}(X,Z),X),Y)
\]
because it adjusts for the first stage propensity score estimation. }{\small\par}

\subsection*{A More General MTE Model}

\subsubsection*{Nonlinear Model with Interaction Effects}

{\small{}We derive results under a more general model for marginal
treatment effects $\theta_{MTE}(X,V)$ that permits interactions between
covariates $X$ and unobserved heterogeneity $V$, and also nonlinear
effects of $V$ of polynomial order $S$. The model considered in
the main text holds under $S=1$ (linear effects) and by setting the
coefficients of the interactions between $X$ and $V$ to zero.}{\small\par}

{\small{}We assume a linear model for these potential outcomes in
terms of the observed covariates $X$, 
\[
Y_{a}=\alpha_{a}+\beta_{a}^{\prime}X+U_{a}
\]
for $a=0,1$, where $\alpha_{a}$ and $\beta_{a}$ are unknown parameters,
and $U_{a}$ are unobserved errors which, without loss of generality,
have mean zero conditional on $X$. }{\small\par}

{\small{}Following previous work \citep{Brinch2017}, we consider
the following assumption on the conditional means of unobserved confounders,
\begin{equation}
E[U_{a}\vert X=x,V=v]=\sum_{s=1}^{S}\zeta_{as}\Big(v^{s}-\frac{1}{s+1}\Big)+\xi_{as}^{\prime}x\Big(v^{s}-\frac{1}{s+1}\Big),\,\,\,\text{for}\,\,a=0,1,
\end{equation}
where $(\zeta_{01}\ldots\zeta_{0S},\zeta_{11}\ldots\zeta_{1S})$ and
$(\xi_{01}\ldots\xi_{0S},\xi_{11}\ldots\xi_{1S})$ are unknown parameters,
and the constant terms are restricted to satisfy $E_{{\cal P}}[U_{0}\vert X]=0$
and $E_{{\cal P}}[U_{1}\vert X]=0$. }{\small\par}

{\small{}Within Equation (1), a popular assumption sets $\xi_{01}=\ldots=\xi_{0S}=0$
and $\xi_{11}=\ldots=\xi_{1S}=0$; this is called }\textsl{\small{}additive
separability}{\small{} \citep{Carneiro2009,Brinch2017}, because it
will mean that the marginal treatment effect is separable in the observed
covariates $X$ and unobserved heterogeneity $V$. It is implied by
$(X,Z)$ being independent of $(U_{0},U_{1},V)$ which, as \citet{Mogstad2018}
describe, would ``nearly elevate $X$ to the status of an instrument,
albeit one that does not need to obey the usual exclusion restriction''.
In our application, we will take the observed covariates $X$ to be
sex, which may satisfy the additive separability assumption. }{\small\par}

{\small{}The advantage of imposing additive separability is that it
allows identification of marginal treatment effects using the collective
variation in $\pi_{{\cal P}}(X,Z)$ and $X$, rather than the variation
in $\pi_{{\cal P}}(X,Z)$ conditional on $X$. In contrast, the more
general model where all parameters may be non-zero allows }\textsl{\small{}interaction
effects}{\small{}, so that the marginal treatment effect may vary
in $V$ in a different way across covariates $X$. }{\small\par}

{\small{}Another important choice is the order of polynomial $S$.
Our analysis in the main text considers the $S=1$ case because observational
studies suggest a linear association between alcohol consumption and
systolic blood pressure \citep{DiFederico2023}, and the individual
parameters of the MTE model have an easier interpretation. As we let
$S$ increase, the underlying specifications for the conditional potential
outcome means $E[Y_{0}\vert X,V]$ and $E[Y_{1}\vert X,V]$ are akin
to a semiparametric model. In particular, higher order polynomials
are assumed to approximate the conditional means $E[U_{0}\vert X,V]$
and $E[U_{1}\vert X,V]$, with the usual caveat that unknown parameters
in more complex models may be difficult to estimate precisely. }{\small\par}

\subsubsection*{Estimands}

{\small{}Under this more general model, the conditional mean of the
observed outcome $m(x,p)=E_{{\cal P}}[Y\vert X=x,\pi_{{\cal P}}(X,Z)=p]$
in Equation (4) of the main text now takes the form $m(x,p)=r(x,p)^{\prime}\gamma$
where $r(x,p)$ is a known vector, and $\boldsymbol{\gamma}$ is an
unknown vector such that 
\[
r(x,p)=\begin{bmatrix}1\\
x\\
p\\
xp\\
p^{2}\\
\vdots\\
p^{S+1}\\
xp^{2}\\
\vdots\\
xp^{S+1}
\end{bmatrix},\,\,\text{and}\,\,\gamma=\begin{bmatrix}\alpha_{0}\\
\beta_{0}\\
\alpha_{1}-\alpha_{0}-\sum_{s=1}^{S}\frac{1}{s+1}(\zeta_{1s}-\zeta_{0s})\\
\text{\ensuremath{\beta_{1}-\beta_{0}-\sum_{s=1}^{S}\frac{1}{s+1}(\xi_{1s}-\xi_{0s})}}\\
\frac{1}{2}(\zeta_{11}-\zeta_{01})\\
\vdots\\
\frac{1}{S+1}(\zeta_{1S}-\zeta_{0S})\\
\frac{1}{2}(\xi_{11}-\xi_{01})\\
\vdots\\
\frac{1}{S+1}(\xi_{1S}-\xi_{0S})
\end{bmatrix}.
\]
}{\small\par}

{\small{}The MTE function is identified as the derivative $\theta_{\text{MTE}}(x,v)=\partial m(x,p)\big/\partial p\vert_{p=v}$,
and in the linear model ($S=1$) the MTE simplifies as $\theta_{\text{MTE}}(x,v)=(\partial r(x,v)\big/\partial v)^{\prime}\gamma$
where $\partial r(x,v)\big/\partial v=(0,0^{\prime},1,x^{\prime},2p,2px^{\prime})^{\prime}$.
This specification of the MTE means that the treatment effect can
vary linearly with covariates $x$ and unobserved heterogeneity $v$,
and the intercept of the linear-in-$v$ MTE slope can vary for subgroups
of covariates (in our empirical example, the MTE intercept may be
different for women than for men). }{\small\par}

{\small{}Inspecting the signs of individual parameters in $\gamma$
may be interesting in simple models, but is less insightful in more
complex models (with a higher order polynomial $S$). Instead, scalar
treatment effect estimands can be more attractive in terms of interpretability.
We state the form of the four target estimands considered in the main
text under this general model. }{\small\par}

{\small{}First, the }\textsl{\small{}average treatment effect}{\small{}
(ATE) effect is 
\begin{equation}
\theta_{\text{ATE}}(\gamma,{\cal P})=E_{{\cal P}}[Y_{1}-Y_{0}]=E_{{\cal P}}[r_{\text{ATE}}(X)]^{\prime}\gamma
\end{equation}
where $r_{\text{ATE}}(X)=(0,0^{\prime},1,X^{\prime},1_{S\times1}^{\prime},\underset{S\text{ times}}{\underbrace{X^{\prime},X^{\prime},\ldots,X^{\prime}}})^{\prime}$.}{\small\par}

{\small{}The }\textsl{\small{}average treatment effect on the treated}{\small{}
(ATT) is 
\begin{equation}
\theta_{\text{ATT}}(\gamma,{\cal P})=E_{{\cal P}}[Y_{1}-Y_{0}\vert A=1]={\cal P}(A=1)^{-1}E_{{\cal P}}[r_{\text{ATT}}(X,Z,{\cal P})]^{\prime}\gamma
\end{equation}
where 
\begin{eqnarray*}
r_{\text{ATT}}(X,Z,{\cal P}) & = & (0,0^{\prime},\pi_{{\cal P}}(X,Z),\pi_{{\cal P}}(X,Z)X^{\prime},\pi_{{\cal P}}(X,Z)^{2},\ldots,\pi_{{\cal P}}(X,Z)^{S+1},\\
 &  & \pi_{{\cal P}}(X,Z)^{2}X^{\prime},\ldots,\pi_{{\cal P}}(X,Z)^{S+1}X^{\prime})^{\prime}.
\end{eqnarray*}
}{\small\par}

{\small{}The }\textsl{\small{}average treatment effect on the untreated}{\small{}
(ATU) is 
\begin{equation}
\theta_{ATU}(\gamma,{\cal P})=E_{{\cal P}}[Y_{1}-Y_{0}\vert A=0]={\cal P}(A=0)^{-1}E_{{\cal P}}[r_{\text{ATU}}(X,Z,{\cal P})]^{\prime}\gamma
\end{equation}
where 
\begin{eqnarray*}
r_{ATU}(X,Z,{\cal P}) & = & (0,0^{\prime},1-\pi_{{\cal P}}(X,Z),(1-\pi_{{\cal P}}(X,Z))X^{\prime},1-\pi_{{\cal P}}(X,Z)^{2},\ldots,1-\pi_{{\cal P}}(X,Z)^{S+1},\\
 &  & (1-\pi_{{\cal P}}(X,Z)^{2})X^{\prime},\ldots,(1-\pi_{{\cal P}}(X,Z)^{S+1})X^{\prime})^{\prime}.
\end{eqnarray*}
}{\small\par}

{\small{}The }\textsl{\small{}average selection on gains}{\small{}
(ASG) is , 
\begin{eqnarray}
\theta_{\text{ASG}}(\gamma,{\cal P}) & = & E_{{\cal P}}[Y_{1}-Y_{0}\vert A=1]-E_{{\cal P}}[Y_{1}-Y_{0}\vert A=0]\nonumber \\
 & = & (E_{{\cal P}}[r_{\text{ATT}}(X,Z,{\cal P})]{\cal P}(A=1)^{-1}-E_{{\cal P}}[r_{\text{ATU}}(X,Z,{\cal P})]{\cal P}(A=0)^{-1})^{\prime}\gamma.
\end{eqnarray}
}{\small\par}

{\small{}\citet{Heckman2005} discussed how some observed estimands
can also be expressed in terms of the vector of MTE parameters $\gamma$.
For example, the IV estimand can be written as 
\begin{equation}
\theta_{\text{IV}}(\gamma,{\cal P})=\text{cov}_{{\cal P}}(A,Z)^{-1}\text{cov}_{{\cal P}}(Y,Z)=\text{cov}_{{\cal P}}(A,Z)^{-1}E_{{\cal P}}[r_{\text{IV}}(X,Z,{\cal P})]^{\prime}\gamma
\end{equation}
where $r_{\text{IV}}(X,Z,{\cal P})=r(X,\pi_{{\cal P}}(X,Z))(Z-E_{{\cal P}}[Z])$.}\\
{\small{}~}{\small\par}

\subsection*{Useful Lemmas}

{\small{}For the following results, $f(o)$ denotes the probability
density function of $O$, and define $f_{t}(o)=t{\cal I}_{\widetilde{o}}(o)+(1-t)f(o)$,
where ${\cal I}_{\widetilde{o}}(o)$ denotes the Dirac delta function
with respect to $\widetilde{o}$; i.e. a function that equals zero
except at $\widetilde{o}$, and that integrates to $1$. }{\small\par}
\begin{lem}
{\small{}Under the mixture model ${\cal P}_{t}=t\widetilde{{\cal P}}+(1-t){\cal P}$
where $\widetilde{{\cal P}}$ denotes a fixed distribution that has
a point mass at a single observation $\widetilde{o}=(\tilde{a},\tilde{x},\tilde{z})$,
\[
\frac{\partial}{\partial t}\pi_{{\cal P}_{t}}(x,z)\Big\vert_{t=0}=\frac{1}{f(x,z)}\big({\cal I}_{(\widetilde{a},\widetilde{x},\widetilde{z})}(1,x,z)-{\cal I}_{(\widetilde{x},\widetilde{z})}(x,z)\pi_{{\cal P}}(x,z)\big).
\]
}{\small\par}
\end{lem}
\begin{proof}
{\small{}For $\pi_{{\cal P}}(x,z)=P(A=1\vert X,Z)=f(1\vert x,z)=f(1,x,z)\big/f(x,z)$,
we can write 
\[
\pi_{{\cal P}_{t}}(x,z)=\frac{f_{t}(1,x,z)}{f_{t}(x,z)}.
\]
Since $f_{t}(1,x,z)=t{\cal I}_{(\widetilde{a},\widetilde{x},\widetilde{z})}(1,x,z)+(1-t)f(1,x,z)$
and $f_{t}(x,z)=t{\cal I}_{(\widetilde{x},\widetilde{z})}(x,z)+(1-t)f(x,z)$,
we have $\frac{\partial}{\partial t}f_{t}(1,x,z)={\cal I}_{(\widetilde{a},\widetilde{x},\widetilde{z})}(1,x,z)-f(1,x,z)$
and $\frac{\partial}{\partial t}f_{t}(x,z)={\cal I}_{(\widetilde{x},\widetilde{z})}(x,z)-f(x,z)$.
Therefore, 
\begin{eqnarray*}
\frac{\partial}{\partial t}\pi_{{\cal P}_{t}}(x,z)\Big\vert_{t=0} & = & \frac{1}{f(x,z)}\cdot\frac{\partial}{\partial t}f_{t}(1,x,z)-\frac{f(1,x,z)}{f(x,z)^{2}}\cdot\frac{\partial}{\partial t}f_{t}(x,z)\\
 & = & \frac{1}{f(x,z)}{\cal I}_{(\widetilde{a},\widetilde{x},\widetilde{z})}(1,x,z)-f(1\vert x,z)-\frac{f(1,x,z)}{f(x,z)^{2}}{\cal I}_{(\tilde{x},\tilde{z})}(x,z)+f(1\vert x,z)\\
 & = & \frac{1}{f(x,z)}\big({\cal I}_{(\tilde{a},\widetilde{x},\widetilde{z})}(1,x,z)-{\cal I}_{(\widetilde{x},\widetilde{z})}(x,z)f(1\vert x,z)\big)
\end{eqnarray*}
as required. }{\small\par}
\end{proof}
{\small{}~}{\small\par}
\begin{lem}
{\small{}Under the mixture model ${\cal P}_{t}=t\widetilde{{\cal P}}+(1-t){\cal P}$
where $\widetilde{{\cal P}}$ denotes a fixed distribution that has
a point mass at a single observation $\widetilde{o}=(\tilde{a},\tilde{x},\tilde{z})$,
\[
\frac{\partial}{\partial t}r(x,\pi_{{\cal P}_{t}}(x,z))\Big\vert_{t=0}=\frac{\partial}{\partial\pi}r(x,\pi_{{\cal P}}(x,z))\cdot\frac{1}{f(x,z)}\big({\cal I}_{(\tilde{a},\widetilde{x},\widetilde{z})}(1,x,z)-{\cal I}_{(\widetilde{x},\widetilde{z})}(x,z)\pi_{{\cal P}}(x,z)\big).
\]
}{\small\par}
\end{lem}
\begin{proof}
{\small{}Follows immediately from the chain rule and Lemma 1.}{\small\par}
\end{proof}
{\small{}~}{\small\par}
\begin{lem}
{\small{}Under the mixture model ${\cal P}_{t}=t\widetilde{{\cal P}}+(1-t){\cal P}$
where $\widetilde{{\cal P}}$ denotes a fixed distribution that has
a point mass at a single observation $\widetilde{o}=(\tilde{a},\tilde{x},\widetilde{y},\tilde{z})$,
\begin{eqnarray*}
\frac{\partial}{\partial t}E_{{\cal P}_{t}}[r(x,\pi_{{\cal P}_{t}}(x,z))y]\Big\vert_{t=0} & = & \frac{\partial}{\partial\pi}r(\widetilde{x},\pi_{{\cal P}}(\widetilde{x},\widetilde{z}))({\cal I}_{\tilde{a}}(1)-\pi_{{\cal P}}(\widetilde{x},\widetilde{z}))E_{{\cal P}}[y\vert\widetilde{x},\widetilde{z}]+r(\widetilde{x},\pi_{{\cal P}}(\widetilde{x},\widetilde{z}))\widetilde{y}\\
 &  & -E_{{\cal P}}[r(x,\pi_{{\cal P}}(x,z))y].
\end{eqnarray*}
}{\small\par}
\end{lem}
\begin{proof}
{\small{}Note that 
\begin{eqnarray*}
\frac{\partial}{\partial t}E_{{\cal P}_{t}}[r(x,\pi_{{\cal P}_{t}}(x,z))y]\Big\vert_{t=0} & = & \frac{\partial}{\partial t}\int r(x,\pi_{{\cal P}_{t}}(x,z))y\cdot f_{t}(x,y,z)\Big\vert_{t=0}\,dxdydz\\
 & = & \int\frac{\partial}{\partial t}r(x,\pi_{{\cal P}_{t}}(x,z))\Big\vert_{t=0}y\cdot f(x,y,z)\,dxdydz\\
 &  & +\int r(x,\pi_{{\cal P}}(x,z))y\cdot\frac{\partial}{\partial t}f_{t}(x,y,z)\Big\vert_{t=0}\,dxdydz\\
 & = & \int\frac{\partial}{\partial\pi}r(x,\pi_{{\cal P}}(x,z))\cdot\frac{1}{f(x,z)}\big({\cal I}_{(\tilde{a},\widetilde{x},\widetilde{z})}(1,x,z)-{\cal I}_{(\widetilde{x},\widetilde{z})}(x,z)\pi_{{\cal P}}(x,z)\big)y\\
 &  & \times f(x,y,z)\,dxdydz\\
 &  & +\int r(x,\pi_{{\cal P}}(x,z))y\cdot{\cal I}_{(\widetilde{x},\widetilde{y},\widetilde{z})}(x,y,z)\,dxdydz-E_{{\cal P}}[r(x,\pi_{{\cal P}}(x,z))y]\\
 & = & \int\frac{\partial}{\partial\pi}r(\widetilde{x},\pi_{{\cal P}}(\widetilde{x},\widetilde{z}))({\cal I}_{\tilde{a}}(1)-\pi_{{\cal P}}(\widetilde{x},\widetilde{z}))y\cdot f(y\vert\widetilde{x},\widetilde{z})\,dy\\
 &  & +r(\widetilde{x},\pi_{{\cal P}}(\widetilde{x},\widetilde{z}))\widetilde{y}-E_{{\cal P}}[r(x,\pi_{{\cal P}}(x,z))y]\\
 & = & \frac{\partial}{\partial\pi}r(\widetilde{x},\pi_{{\cal P}}(\widetilde{x},\widetilde{z}))({\cal I}_{\tilde{a}}(1)-\pi_{{\cal P}}(\widetilde{x},\widetilde{z}))E_{{\cal P}}[y\vert\widetilde{x},\widetilde{z}]+r(\widetilde{x},\pi_{{\cal P}}(\widetilde{x},\widetilde{z}))\widetilde{y}\\
 &  & -E_{{\cal P}}[r(x,\pi_{{\cal P}}(x,z))y]
\end{eqnarray*}
}{\small\par}

{\small{}since ${\cal I}_{\tilde{a}}(1)$, $\pi_{{\cal P}}(\tilde{x},\tilde{z})$,
and $\frac{\partial}{\partial\pi}r(\tilde{x},\pi_{{\cal P}}(\tilde{x},\tilde{z}))$
are constant, and where the third equality follows by Lemma 2.}{\small\par}
\end{proof}
{\small{}~}{\small\par}
\begin{lem}
{\small{}Under the mixture model ${\cal P}_{t}=t\widetilde{{\cal P}}+(1-t){\cal P}$
where $\widetilde{{\cal P}}$ denotes a fixed distribution that has
a point mass at a single observation $\widetilde{o}=(\tilde{a},\tilde{x},\tilde{z})$,
\begin{eqnarray*}
\frac{\partial}{\partial t}E_{{\cal P}_{t}}[r(x,\pi_{{\cal P}_{t}}(x,z))r(x,\pi_{{\cal P}_{t}}(x,z))^{\prime}]\Big\vert_{t=0} & = & 2\frac{\partial}{\partial\pi}r(\widetilde{x},\pi_{{\cal P}}(\widetilde{x},\widetilde{z}))({\cal I}_{\tilde{a}}(1)-\pi_{{\cal P}}(\widetilde{x},\widetilde{z}))r(\widetilde{x},\pi_{{\cal P}}(\widetilde{x},\widetilde{z}))^{\prime}\\
 &  & +r(\widetilde{x},\pi_{{\cal P}}(\widetilde{x},\widetilde{z}))r(\widetilde{x},\pi_{{\cal P}}(\widetilde{x},\widetilde{z}))^{\prime}\\
 &  & -E_{{\cal P}}[r(x,\pi_{{\cal P}}(x,z))r(x,\pi_{{\cal P}}(x,z))^{\prime}].
\end{eqnarray*}
}{\small\par}
\end{lem}
\begin{proof}
{\small{}Note that 
\begin{align*}
 & \frac{\partial}{\partial t}E_{{\cal P}_{t}}[r(x,\pi_{{\cal P}_{t}}(x,z))r(x,\pi_{{\cal P}_{t}}(x,z))^{\prime}]\Big\vert_{t=0}\\
 & =\frac{\partial}{\partial t}\int\Big(r(x,\pi_{{\cal P}_{t}}(x,z))r(x,\pi_{{\cal P}_{t}}(x,z))^{\prime}\cdot f_{t}(x,z)\,dxdz\\
 & =2\int\frac{\partial}{\partial\pi}r(x,\pi_{{\cal P}}(x,z))\cdot\frac{1}{f(x,z)}({\cal I}_{\tilde{a}}(1)-{\cal I}_{(\widetilde{x},\widetilde{z})}(x,z)\pi_{{\cal P}}(x,z))r(x,\pi_{{\cal P}}(x,z))^{\prime}\cdot f(x,z)\,dxdz\\
 & +r(\widetilde{x},\pi_{{\cal P}}(\widetilde{x},\widetilde{z}))r(\widetilde{x},\pi_{{\cal P}}(\widetilde{x},\widetilde{z}))^{\prime}-E_{{\cal P}}[r(x,\pi_{{\cal P}}(x,z))r(x,\pi_{{\cal P}}(x,z))^{\prime}]\\
 & =2\Big(\frac{\partial}{\partial\pi}r(\widetilde{x},\pi_{{\cal P}}(\widetilde{x},\widetilde{z}))\Big)({\cal I}_{\tilde{a}}(1)-\pi_{{\cal P}}(\widetilde{x},\widetilde{z}))r(\widetilde{x},\pi_{{\cal P}}(\widetilde{x},\widetilde{z}))^{\prime}\\
 & +r(\widetilde{x},\pi_{{\cal P}}(\widetilde{x},\widetilde{z}))r(\widetilde{x},\pi_{{\cal P}}(\widetilde{x},\widetilde{z}))^{\prime}-E_{{\cal P}}[r(x,\pi_{{\cal P}}(x,z))r(x,\pi_{{\cal P}}(x,z))^{\prime}]
\end{align*}
where the second equality follows by Lemma 2. }{\small\par}
\end{proof}
{\small{}~}{\small\par}
\begin{lem}
{\small{}$E_{{\cal P}}[Y\vert X,Z]=E_{{\cal P}}[Y\vert X,\pi_{{\cal P}}(X,Z)]$.}{\small\par}
\end{lem}
\begin{proof}
{\small{}Note that by the law of iterated expectations, 
\begin{eqnarray*}
E_{{\cal P}}[Y\vert X,Z] & = & E[AY_{1}+(1-A)Y_{0}\vert X,Z]\\
 & = & E[AY_{1}\vert X,Z]+E[(1-A)Y_{0}\vert X,Z]\\
 & = & E_{A\vert X,Z}[E[Y_{1}\vert X,Z,A]\cdot A\vert X,Z]+E_{A\vert X,Z}[E[Y_{0}\vert X,Z,A]\cdot(1-A)\vert X,Z]\\
 & = & E[Y_{1}\vert X,Z,A=1]\cdot{\cal P}(A=1\vert X,Z)+E[Y_{0}\vert X,Z,A=0]\cdot(1-{\cal P}(A=1\vert X,Z)).
\end{eqnarray*}
For a mean conditional on the propensity score, we can write 
\begin{eqnarray*}
E_{{\cal P}}[A\vert X,\pi_{{\cal P}}(X,Z)] & = & E_{Z\vert X,\pi_{{\cal P}}(X,Z)}\big[E[A\vert X,\pi_{{\cal P}}(X,Z),Z]\vert X,\pi_{{\cal P}}(X,Z)\big]\\
 & = & E_{Z\vert X,\pi_{{\cal P}}(X,Z)}\big[E_{{\cal P}}[A\vert X,Z]\vert X,\pi_{{\cal P}}(X,Z)\big]\\
 & = & E_{Z\vert X,\pi_{{\cal P}}(X,Z)}\big[\pi_{{\cal P}}(X,Z)\vert X,\pi_{{\cal P}}(X,Z)\big]\\
 & = & \pi_{{\cal P}}(X,Z)\\
 & = & {\cal P}(A=1\vert X,Z)
\end{eqnarray*}
where the first equality follows by the law of iterated expectations,
the second equality follows by $\pi_{{\cal P}}(X,Z)$ being known
given knowledge of $(X,Z)$, and the third and fourth equalities follow
by the definition of the propensity score. Therefore, ${\cal P}(A=1\vert X,Z)={\cal P}(A=1\vert X,\pi_{{\cal P}}(X,Z))$. }{\small\par}

{\small{}Now, by assumption, $Y_{1}\independent Z\vert X$. Therefore,
$Y_{1}\independent(X,Z)\vert X$. Then, $Y_{1}\independent(X,Z)\vert(X,\pi_{{\cal P}}(X,Z))$
since $(X,\pi_{{\cal P}}(X,Z))$ gives no information on $Y$ if $(X,Z)\vert X$
does not. Finally, $Y_{1}\independent(X,Z,A=1)\vert(X,\pi_{{\cal P}}(X,Z),A=1)$.
By identical arguments, $Y_{0}\independent(X,Z,A=0)\vert(X,\pi_{{\cal P}}(X,Z),A=0)$.
Hence, $E[Y_{1}\vert X,Z,A=1]=E[Y_{1}\vert X,\pi_{{\cal P}}(X,Z),A=1]$,
and likewise, $E[Y_{0}\vert X,Z,A=0]=E[Y_{0}\vert X,\pi_{{\cal P}}(X,Z),A=0]$. }{\small\par}

{\small{}Overall, we have shown that 
\begin{eqnarray*}
E_{{\cal P}}[Y\vert X,Z] & = & E[Y_{1}\vert X,Z,A=1]\cdot{\cal P}(A=1\vert X,Z)+E[Y_{0}\vert X,Z,A=0]\cdot(1-{\cal P}(A=1\vert X,Z))\\
 & = & E[Y_{1}\vert X,\pi_{{\cal P}}(X,Z),A=1]\cdot{\cal P}(A=1\vert X,\pi_{{\cal P}}(X,Z))\\
 &  & +E[Y_{0}\vert X,\pi_{{\cal P}}(X,Z),A=0]\cdot{\cal P}(A=0\vert X,\pi_{{\cal P}}(X,Z))\\
 & = & E_{{\cal P}}[Y\vert X,\pi_{{\cal P}}(X,Z)]
\end{eqnarray*}
as required.}{\small\par}
\end{proof}
{\small{}~}{\small\par}
\begin{lem}
{\small{}Under the mixture model ${\cal P}_{t}=t\widetilde{{\cal P}}+(1-t){\cal P}$
where $\widetilde{{\cal P}}$ denotes a fixed distribution that has
a point mass at a single observation $\widetilde{o}=(\tilde{a},\tilde{x},\widetilde{y},\tilde{z})$,
the efficient influence function for $\boldsymbol{\gamma}({\cal P})$
is 
\begin{eqnarray*}
\phi(\widetilde{o},{\cal P}) & = & -E_{{\cal P}}[r(x,\pi_{{\cal P}}(x,z))r(x,\pi_{{\cal P}}(x,z))^{\prime}]^{-1}\frac{\partial}{\partial\pi}r(\widetilde{x},\pi_{{\cal P}}(\widetilde{x},\widetilde{z}))(\widetilde{a}-\pi_{{\cal P}}(\widetilde{x},\widetilde{z}))r(\widetilde{x},\pi_{{\cal P}}(\widetilde{x},\widetilde{z}))^{\prime}\boldsymbol{\gamma}({\cal P})\\
 &  & -E_{{\cal P}}[r(x,\pi_{{\cal P}}(x,z))r(x,\pi_{{\cal P}}(x,z))^{\prime}]^{-1}r(\widetilde{x},\pi_{{\cal P}}(\widetilde{x},\widetilde{z}))r(\widetilde{x},\pi_{{\cal P}}(\widetilde{x},\widetilde{z}))^{\prime}\boldsymbol{\gamma}({\cal P})\\
 &  & +E_{{\cal P}}[r(x,\pi_{{\cal P}}(x,z))r(x,\pi_{{\cal P}}(x,z))^{\prime}]^{-1}r(\widetilde{x},\pi_{{\cal P}}(\widetilde{x},\widetilde{z}))\widetilde{y}.
\end{eqnarray*}
}{\small\par}
\end{lem}
\begin{proof}
{\small{}Note that
\begin{eqnarray*}
\phi(\widetilde{o},{\cal P}) & = & \frac{\partial}{\partial t}\gamma_{0}({\cal P}_{t})\Big\vert_{t=0}\\
 & = & -E_{{\cal P}}[r(\pi_{{\cal P}}(x,z),x)r(\pi_{{\cal P}}(x,z),x)^{\prime}]^{-1}\frac{\partial}{\partial t}E_{{\cal P}_{t}}[r(\pi_{{\cal P}_{t}}(x,z),x)r(\pi_{{\cal P}_{t}}(x,z),x)^{\prime}]\Big\vert_{t=0}\\
 &  & \times E_{{\cal P}}[r(\pi_{{\cal P}}(x,z),x)r(\pi_{{\cal P}}(x,z),x)^{\prime}]^{-1}E_{{\cal P}}[r(\pi_{{\cal P}}(x,z),x)y]\\
 &  & +E_{{\cal P}}[r(\pi_{{\cal P}}(x,z),x)r(\pi_{{\cal P}}(x,z),x)^{\prime}]^{-1}\frac{\partial}{\partial t}E_{{\cal P}_{t}}[r(\pi_{{\cal P}_{t}}(x,z),x)y]\Big\vert_{t=0}
\end{eqnarray*}
which leads to the required result after substituting in the results
from Lemmas 3, 4 and 5. }{\small\par}
\end{proof}
~

\subsubsection*{Proof of Proposition 1}

{\small{}Using Lemma 6, setting the sample average of the efficient
influence function to zero, we have 
\begin{eqnarray*}
0 & = & E_{\widehat{{\cal P}}}[\phi(o,\widehat{{\cal P}})]\\
 & = & -E_{{\cal \widehat{P}}}[r(x,\pi_{{\cal P}}(x,z))r(x,\pi_{{\cal P}}(x,z))^{\prime}]^{-1}E_{\widehat{{\cal P}}}\Big[\frac{\partial}{\partial\pi}r(x,\pi_{{\cal P}}(x,z))(a-\pi_{{\cal P}}(x,z))r(x,\pi_{{\cal P}}(x,z))^{\prime}\Big]\gamma(\widehat{{\cal P}})\\
 &  & -\gamma(\widehat{{\cal P}})+E_{\widehat{{\cal P}}}[r(x,\pi_{{\cal P}}(x,z))r(x,\pi_{{\cal P}}(x,z))^{\prime}]^{-1}E_{\widehat{{\cal P}}}[r(x,\pi_{{\cal P}}(x,z))y].
\end{eqnarray*}
Solving for $\widetilde{\gamma}=\gamma(\widehat{{\cal P}})$, the
efficient estimator $\widetilde{\gamma}$ has a convenient closed
form expression
\[
\widetilde{\gamma}=(\Omega_{{\cal \widehat{P}}}+\Gamma_{\widehat{{\cal P}}})^{-1}\Upsilon_{\widehat{{\cal P}}}
\]
where $\Omega_{{\cal \widehat{P}}}$, $\Gamma_{\widehat{{\cal P}}}$,
and $\Upsilon_{\widehat{{\cal P}}}$ are 
\begin{eqnarray*}
\Omega_{{\cal \widehat{P}}} & = & E_{{\cal P}}[r(x,\pi_{{\cal P}}(x,z))r(x,\pi_{{\cal P}}(x,z))^{\prime}]\\
\Gamma_{\widehat{{\cal P}}} & = & E_{{\cal P}}\big[\frac{\partial}{\partial\pi}r(x,\pi_{{\cal P}}(x,z))(a-\pi_{{\cal P}}(x,z))r(x,\pi_{{\cal P}}(x,z))^{\prime}]\\
\Upsilon_{\widehat{{\cal P}}} & = & E_{{\cal P}}[r(x,\pi_{{\cal P}}(x,z))y].
\end{eqnarray*}
The asymptotic variance of the efficient estimator $\widetilde{\gamma}$
can be calculated as $E_{\widehat{{\cal P}}}[\phi(o,{\cal P})\phi(o,{\cal P})^{\prime}]$.
From Lemma 6, this variance expression takes the form $E_{\widehat{{\cal P}}}[\phi(o,{\cal P})\phi(o,{\cal P})^{\prime}]=\Omega_{{\cal P}}^{-1}\Phi_{{\cal P}}\Omega_{{\cal P}}^{-1\prime}$
where 
\[
\Phi_{{\cal P}}=E_{{\cal P}}[(g(o,\gamma,{\cal P})+\psi(o,\gamma,{\cal P}))(g(o,\gamma,{\cal P})+\psi(o,\gamma,{\cal P}))^{\prime}]
\]
 for 
\[
g(o,\gamma,{\cal P})=r(x,\pi_{{\cal P}}(x,z))(y-r(x,\pi_{{\cal P}}(x,z))^{\prime}\gamma)
\]
 and 
\[
\psi(o,\gamma,{\cal P})=\frac{\partial}{\partial\pi}r(x,\pi_{{\cal P}}(x,z))(a-\pi_{{\cal P}}(x,z))r(x,\pi_{{\cal P}}(x,z))^{\prime}\gamma
\]
and where $\Omega_{{\cal P}}$ and $\Gamma_{{\cal P}}$ are defined
in the main text. }{\small\par}

{\small{}Now we need to show that the conventional (OLS) estimator
$\widehat{\gamma}=\Omega_{{\cal \widehat{P}}}^{-1}cov_{\widehat{{\cal P}}}(r(\pi_{\widehat{{\cal P}}}(x,z),x),y)$
has the same asymptotic variance as $\widetilde{\gamma}$. To derive
the asymptotic variance of $\widehat{\gamma}$, we start by deriving
the so-called first stage influence function (FSIF; \citealt{Chernozhukov2022}),
of the OLS moment condition $E_{{\cal P}}[g(o,\gamma,{\cal P})]$.
where $g(o,\gamma,{\cal P})=r(x,\pi_{{\cal P}}(x,z))(y-r(x,\pi_{{\cal P}}(x,z))^{\prime}\gamma)$. }{\small\par}

{\small{}The FSIF term, which we denote $\text{FSIF}(o,\gamma,{\cal P})$,
relates to the uncertainty due to first-stage propensity score estimation,
and under standard regularity conditions, the first-order asymptotic
expansion of the conventional estimator $\widehat{\gamma}$ is $\sqrt{n}(\widehat{\gamma}-\gamma)=\Omega_{{\cal P}}^{-1}\sqrt{n}E_{\widehat{{\cal P}}}\big[g(O,\gamma,{\cal P})+\text{FSIF}(o,\gamma,{\cal P})\big]+o_{{\cal P}}(1)$. }{\small\par}

{\small{}Using the same point mass contamination strategy used above,
we calculate the FSIF at observation $\widetilde{o}=(\widetilde{a},\widetilde{x},\widetilde{y},\widetilde{z})$
as the Gateaux derivative 
\[
\text{FSIF}(o,\gamma,{\cal P})=\frac{\partial}{\partial t}E_{{\cal P}}[g(\widetilde{o},\gamma,{\cal P}_{t})]\Big\vert_{t=0}.
\]
Calculating this derivative, we have 
\begin{eqnarray*}
\text{FSIF}(o,\gamma,{\cal P}) & = & \frac{\partial}{\partial t}\int r(x,\pi_{{\cal P}_{t}}(x,z))(y-r(x,\pi_{{\cal P}_{t}}(x,z))^{\prime}\gamma)\cdot f(x,y,z)\,dxdydz\Big\vert_{t=0}\\
 & = & \frac{\partial}{\partial t}E_{{\cal P}}[r(x,\pi_{{\cal P}_{t}}(x,z))y]\Big\vert_{t=0}-\Big(\frac{\partial}{\partial t}E_{{\cal P}}[r(x,\pi_{{\cal P}_{t}}(x,z))r(x,\pi_{{\cal P}_{t}}(x,z))^{\prime}]\Big\vert_{t=0}\Big)\gamma\\
 & = & \frac{\partial}{\partial\pi}r(\widetilde{x},\pi_{{\cal P}}(\widetilde{x},\widetilde{z}))(\widetilde{a}-\pi_{{\cal P}}(\widetilde{x},\widetilde{z}))E_{{\cal P}}[y\vert\widetilde{x},\widetilde{z}]\\
 &  & -2\Big(\frac{\partial}{\partial\pi}r(\widetilde{x},\pi_{{\cal P}}(\widetilde{x},\widetilde{z}))\Big)(\tilde{a}-\pi_{{\cal P}}(\widetilde{x},\widetilde{z}))r(\widetilde{x},\pi_{{\cal P}}(\widetilde{x},\widetilde{z}))^{\prime}\gamma\\
 & = & -\Big(\frac{\partial}{\partial\pi}r(\widetilde{x},\pi_{{\cal P}}(\widetilde{x},\widetilde{z}))\Big)(\tilde{a}-\pi_{{\cal P}}(\widetilde{x},\widetilde{z}))r(\widetilde{x},\pi_{{\cal P}}(\widetilde{x},\widetilde{z}))^{\prime}\gamma\\
 & = & \psi(o,\gamma,{\cal P})
\end{eqnarray*}
using Lemma 5, and where the third equality uses similar arguments
used in Proofs of Lemmas 3 and 4. Thus, given $\sqrt{n}(\widehat{\gamma}-\gamma)=\Omega_{{\cal P}}^{-1}\sqrt{n}E_{\widehat{{\cal P}}}\big[g(O,\gamma,{\cal P})+\text{FSIF}(o,\gamma,{\cal P})\big]+o_{{\cal P}}(1)$,
we have that the asymptotic variance of $\widehat{\gamma}$ is equal
to $\Omega_{{\cal P}}^{-1}\Phi_{{\cal P}}\Omega_{{\cal P}}^{-1\prime}$,
which is the same asymptotic variance as $\widetilde{\gamma}$. \hfill{}$\square$}\\
{\small{}~}{\small\par}

\subsubsection*{Proof of Proposition 2}

\textsc{(i) IV estimand}

{\small{}By Table II of \citet[p.1597]{Mogstad2018a}, we can express
the IV estimand $cov(A,Z)^{-1}cov(Y,Z)$ as a weighted average of
mean response functions $E[Y_{1}\vert X,V]$ and $E[Y_{0}\vert X,V]$.
In our model, these mean response functions take the form $E[Y_{0}\vert X,V]=\alpha_{0}+\beta_{0}^{\prime}X+E[U_{0}\vert X,V]$
and $E[Y_{1}\vert X,V]=\alpha_{1}+\beta_{1}^{\prime}X+E[U_{1}\vert X,V]$.
By Equation (1), $E[U_{0}\vert X,V]=\sum_{s=1}^{S}\zeta_{0s}\Big(V^{s}-\frac{1}{s+1}\Big)+\xi_{0s}^{\prime}X\Big(V^{s}-\frac{1}{s+1}\Big)$
and $E[U_{1}\vert X,V]=\sum_{s=1}^{S}\zeta_{1s}\Big(V^{s}-\frac{1}{s+1}\Big)+\xi_{1s}^{\prime}X\Big(V^{s}-\frac{1}{s+1}\Big)$.
Hence, using the formula given in Table II of \citet[p.1597]{Mogstad2018a},
we can write 
\begin{eqnarray*}
cov(Y,Z) & = & E\bigg[\int_{p}^{1}\Big\{\alpha_{0}+\beta_{0}^{\prime}X+\sum_{s=1}^{S}\zeta_{0s}\Big(v^{s}-\frac{1}{s+1}\Big)+\xi_{0s}^{\prime}X\Big(v^{s}-\frac{1}{s+1}\Big)\Big\} dv\cdot(Z-E[Z])\bigg]\\
 &  & +E\bigg[\int_{0}^{p}\Big\{\alpha_{1}+\beta_{1}^{\prime}X+\sum_{s=1}^{S}\zeta_{1s}\Big(v^{s}-\frac{1}{s+1}\Big)+\xi_{1s}^{\prime}X\Big(v^{s}-\frac{1}{s+1}\Big)\Big\} dv\cdot(Z-E[Z])\bigg]\\
 & = & \alpha_{0}E[Z-E[Z]]+\beta_{0}^{\prime}E[X(Z-E[Z])]+\Big(\alpha_{1}-\alpha_{0}-\sum_{s=1}^{S}\frac{1}{s+1}(\zeta_{1s}-\zeta_{0s})\Big)E[p(Z-E[Z])]\\
 &  & +\Big(\beta_{1}-\beta_{0}-\sum_{s=1}^{S}\frac{1}{s+1}(\xi_{1s}-\xi_{0s})\Big)^{\prime}E[Xp(Z-E[Z])]\\
 &  & +\sum_{s=1}^{S}\frac{1}{s+1}(\zeta_{1s}-\zeta_{0s})E[p^{s+1}(Z-E[Z])]+\sum_{s=1}^{S}\frac{1}{s+1}(\xi_{1s}-\xi_{0s})E[Xp^{s+1}(Z-E[Z])].
\end{eqnarray*}
Therefore, $\theta_{IV}({\cal P})=cov_{{\cal P}}(A,Z)^{-1}cov_{{\cal P}}(Y,Z)=cov_{{\cal P}}(A,Z)^{-1}E_{{\cal P}}[r(X,\pi_{{\cal P}}(X,Z))(Z-E_{{\cal P}}[Z])]^{\prime}\gamma({\cal P})$. }{\small\par}

{\small{}The efficient influence function of $\theta_{IV}({\cal P})$
can be calculated by perturbing $\theta_{IV}({\cal P})$ in the direction
$\widetilde{{\cal P}}$ using the mixture model ${\cal P}_{t}=t\widetilde{{\cal P}}+(1-t){\cal P}$,
where $\widetilde{{\cal P}}$ denotes a fixed distribution that has
a point mass at a single observation $\widetilde{o}=(\tilde{y},\tilde{a},\tilde{x},\tilde{z})$.
The efficient influence function at $\tilde{o}$ is 
\begin{eqnarray}
\phi_{IV}(\tilde{o},{\cal P}) & = & \frac{\partial}{\partial t}\theta_{IV}({\cal P}_{t})\Big\vert_{t=0}\nonumber \\
 & = & \Big(\frac{\partial}{\partial t}E_{{\cal P}_{t}}[r(X,\pi_{{\cal P}}(X,Z))(Z-E_{{\cal P}}[Z])]\Big\vert_{t=0}\Big)^{\prime}\gamma({\cal P})\cdot cov_{{\cal P}}(A,Z)^{-1}\nonumber \\
 &  & +E_{{\cal P}}[r(X,\pi_{{\cal P}}(X,Z))(Z-E_{{\cal P}}[Z])]^{\prime}\cdot\frac{\partial}{\partial t}\gamma({\cal P}_{t})\Big\vert_{t=0}\cdot cov_{{\cal P}}(A,Z)^{-1}\nonumber \\
 &  & +E_{{\cal P}}[r(X,\pi_{{\cal P}}(X,Z))(Z-E_{{\cal P}}[Z])]^{\prime}\gamma({\cal P})\cdot\frac{\partial}{\partial t}\big(cov_{{\cal P}_{t}}(A,Z)^{-1}\big)\Big\vert_{t=0}.
\end{eqnarray}
We will deal with each of the three terms on the RHS of Equation (7)
in turn. First, 
\begin{eqnarray*}
\frac{\partial}{\partial t}E_{{\cal P}_{t}}[r(X,\pi_{{\cal P}}(X,Z))(Z-E_{{\cal P}}[Z])]\Big\vert_{t=0} & = & r(\tilde{x},\pi_{{\cal P}}(\tilde{x},\tilde{z}))(\tilde{z}-E_{{\cal P}}[Z])-E_{{\cal P}}[r(X,\pi_{{\cal P}}(X,Z))(Z-E_{{\cal P}}[Z])]\\
 &  & +E_{{\cal P}}\Big[\frac{\partial}{\partial t}r(X,\pi_{{\cal P}_{t}}(X,Z))(Z-E_{{\cal P}_{t}}[Z])\Big\vert_{t=0}\Big].
\end{eqnarray*}
Now, 
\begin{eqnarray*}
E_{{\cal P}}\Big[\frac{\partial}{\partial t}r(X,\pi_{{\cal P}_{t}}(X,Z))(Z-E_{{\cal P}_{t}}[Z])\Big\vert_{t=0}\Big] & = & \int\frac{\partial}{\partial t}r(x,\pi_{{\cal P}_{t}}(x,z))(z-E_{{\cal P}_{t}}[Z])\Big\vert_{t=0}\cdot f(x,z)\,dxdz\\
 & = & \int\frac{\partial}{\partial\pi}r(x,\pi_{{\cal P}}(x,z))({\cal I}_{\tilde{a}}(1)-\pi_{{\cal P}}(x,z)){\cal I}_{(\widetilde{x},\widetilde{z})}(x,z)(z-E_{{\cal P}}[Z])\,dxdz\\
 &  & -\int r(x,\pi_{{\cal P}}(x,z))(\tilde{z}-E_{{\cal P}}[Z])f(x,z)\,dxdz\\
 & = & \frac{\partial}{\partial\pi}r(\tilde{x},\pi_{{\cal P}}(\tilde{x},\tilde{z}))(\tilde{a}-\pi_{{\cal P}}(\tilde{x},\tilde{z}))(\tilde{z}-E_{{\cal P}}[Z])\\
 &  & -E_{{\cal P}}[r(X,\pi_{{\cal P}}(X,Z))](\tilde{z}-E_{{\cal P}}[Z]),
\end{eqnarray*}
where the second equality follows by using Lemma 2, and noting that
$\partial E_{{\cal P}_{t}}[Z]/\partial t=\int z({\cal I}_{\tilde{z}}(z)-f(z))dz=\tilde{z}-E_{{\cal P}}[Z]$. }{\small\par}

{\small{}Overall, 
\begin{eqnarray*}
\frac{\partial}{\partial t}E_{{\cal P}_{t}}[r(X,\pi_{{\cal P}}(X,Z))(Z-E_{{\cal P}}[Z])]\Big\vert_{t=0} & = & r(\tilde{x},\pi_{{\cal P}}(\tilde{x},\tilde{z}))(\tilde{z}-E_{{\cal P}}[Z])-E_{{\cal P}}[r(X,\pi_{{\cal P}}(X,Z))(Z-E_{{\cal P}}[Z])]\\
 &  & +\frac{\partial}{\partial\pi}r(\tilde{x},\pi_{{\cal P}}(\tilde{x},\tilde{z}))(\tilde{a}-\pi_{{\cal P}}(\tilde{x},\tilde{z}))(\tilde{z}-E_{{\cal P}}[Z])\\
 &  & -E_{{\cal P}}[r(X,\pi_{{\cal P}}(X,Z))](\tilde{z}-E_{{\cal P}}[Z]).
\end{eqnarray*}
}{\small\par}

{\small{}For the second term on the RHS of Equation (7), Lemma 6 tells
us 
\begin{eqnarray*}
\frac{\partial}{\partial t}\gamma({\cal P}_{t})\Big\vert_{t=0} & = & -\Omega_{{\cal P}}^{-1}\frac{\partial}{\partial\pi}r(\widetilde{x},\pi_{{\cal P}}(\widetilde{x},\widetilde{z}))(\widetilde{a}-\pi_{{\cal P}}(\widetilde{x},\widetilde{z}))r(\widetilde{x},\pi_{{\cal P}}(\widetilde{x},\widetilde{z}))^{\prime}\gamma({\cal P})\\
 &  & -\Omega_{{\cal P}}^{-1}r(\widetilde{x},\pi_{{\cal P}}(\widetilde{x},\widetilde{z}))r(\widetilde{x},\pi_{{\cal P}}(\widetilde{x},\widetilde{z}))^{\prime}\gamma({\cal P})+\Omega_{{\cal P}}^{-1}r(\widetilde{x},\pi_{{\cal P}}(\widetilde{x},\widetilde{z}))\widetilde{y}.
\end{eqnarray*}
Therefore, the second term on the RHS of Equation (7) is equal to
\begin{align*}
cov_{{\cal P}}(A,Z)^{-1}E_{{\cal P}}[r(X,\pi_{{\cal P}}(X,Z))(Z-E_{{\cal P}}[Z])]^{\prime}\Omega_{{\cal P}}^{-1}r(\widetilde{x},\pi_{{\cal P}}(\widetilde{x},\widetilde{z}))\big(\widetilde{y}-r(\widetilde{x},\pi_{{\cal P}}(\widetilde{x},\widetilde{z}))^{\prime}\gamma({\cal P})\big)\\
-cov_{{\cal P}}(A,Z)^{-1}E_{{\cal P}}[r(X,\pi_{{\cal P}}(X,Z))(Z-E_{{\cal P}}[Z])]^{\prime}\Omega_{{\cal P}}^{-1}\frac{\partial}{\partial\pi}r(\widetilde{x},\pi_{{\cal P}}(\widetilde{x},\widetilde{z}))(\widetilde{a}-\pi_{{\cal P}}(\widetilde{x},\widetilde{z}))r(\widetilde{x},\pi_{{\cal P}}(\widetilde{x},\widetilde{z}))^{\prime}\gamma({\cal P}).
\end{align*}
For the third and final term on the RHS of Equation (7), note that
\begin{eqnarray*}
\frac{\partial}{\partial t}E_{{\cal P}_{t}}[(A-E_{{\cal P}_{t}}[A])(Z-E_{{\cal P}_{t}}[Z])]\Big\vert_{t=0} & = & ({\cal I}_{\tilde{a}}(1)-E_{{\cal P}}[A])(\tilde{z}-E_{{\cal P}}[Z])-cov_{{\cal P}}(A,Z)\\
 &  & +E_{{\cal P}}\Big[\frac{\partial}{\partial t}(A-E_{{\cal P}_{t}}[A])(Z-E_{{\cal P}_{t}}[Z])\Big\vert_{t=0}\Big]\\
 & = & ({\cal I}_{\tilde{a}}(1)-E_{{\cal P}}[A])(\tilde{z}-E_{{\cal P}}[Z])-cov_{{\cal P}}(A,Z).
\end{eqnarray*}
Hence, the third term on the RHS of Equation (7) is 
\begin{align*}
-E_{{\cal P}}[r(X,\pi_{{\cal P}}(X,Z))(Z-E_{{\cal P}}[Z])]^{\prime}\gamma({\cal P})\cdot cov_{{\cal P}}(A,Z)^{-2}\Big\{({\cal I}_{\tilde{a}}(1)-E_{{\cal P}}[A])(\tilde{z}-E_{{\cal P}}[Z])-cov_{{\cal P}}(A,Z)\Big\}\\
=-E_{{\cal P}}[r(X,\pi_{{\cal P}}(X,Z))(Z-E_{{\cal P}}[Z])]^{\prime}\gamma({\cal P})\cdot cov_{{\cal P}}(A,Z)^{-2}({\cal I}_{\tilde{a}}(1)-E_{{\cal P}}[A])(\tilde{z}-E_{{\cal P}}[Z])+\theta_{IV}({\cal P})
\end{align*}
Combining the above results, we calculate the efficient influence
function at $\tilde{o}$ as 
\begin{eqnarray*}
\phi_{IV}(\tilde{o},{\cal P}) & = & cov_{{\cal P}}(A,Z)^{-1}\big(r(\tilde{x},\pi_{{\cal P}}(\tilde{x},\tilde{z}))(\tilde{z}-E_{{\cal P}}[Z])\big)^{\prime}\gamma({\cal P})-\theta_{IV}({\cal P})\\
 &  & +cov_{{\cal P}}(A,Z)^{-1}\Big(\frac{\partial}{\partial\pi}r(\tilde{x},\pi_{{\cal P}}(\tilde{x},\tilde{z}))(\tilde{a}-\pi_{{\cal P}}(\tilde{x},\tilde{z}))(\tilde{z}-E_{{\cal P}}[Z])\Big)^{\prime}\gamma({\cal P})\\
 &  & -cov_{{\cal P}}(A,Z)^{-1}(\tilde{z}-E_{{\cal P}}[Z])E_{{\cal P}}[r(X,\pi_{{\cal P}}(X,Z))]^{\prime}\gamma({\cal P})\\
 &  & +cov_{{\cal P}}(A,Z)^{-1}E_{{\cal P}}[r(X,\pi_{{\cal P}}(X,Z))(Z-E_{{\cal P}}[Z])]^{\prime}\Omega_{{\cal P}}^{-1}r(\widetilde{x},\pi_{{\cal P}}(\widetilde{x},\widetilde{z}))\big(\widetilde{y}-r(\widetilde{x},\pi_{{\cal P}}(\widetilde{x},\widetilde{z}))^{\prime}\boldsymbol{\gamma}({\cal P})\big)\\
 &  & -cov_{{\cal P}}(A,Z)^{-1}E_{{\cal P}}[r(X,\pi_{{\cal P}}(X,Z))(Z-E_{{\cal P}}[Z])]^{\prime}\Omega_{{\cal P}}^{-1}\frac{\partial}{\partial\pi}r(\widetilde{x},\pi_{{\cal P}}(\widetilde{x},\widetilde{z}))(\widetilde{a}-\pi_{{\cal P}}(\widetilde{x},\widetilde{z}))r(\widetilde{x},\pi_{{\cal P}}(\widetilde{x},\widetilde{z}))^{\prime}\boldsymbol{\gamma}({\cal P})\\
 &  & -cov_{{\cal P}}(A,Z)^{-1}({\cal I}_{\tilde{a}}(1)-E_{{\cal P}}[A])(\tilde{z}-E_{{\cal P}}[Z])\cdot\theta_{IV}({\cal P})+\theta_{IV}({\cal P}).
\end{eqnarray*}
Setting the empirical mean of $\phi_{IV}(O,{\cal P})$ to zero, we
find an efficient estimator as 
\begin{eqnarray*}
\theta_{IV}(\widehat{{\cal P}}) & = & cov_{\widehat{{\cal P}}}(A,Z)^{-1}E_{\widehat{{\cal P}}}[r(X,\pi_{\widehat{{\cal P}}}(X,Z))(Z-E_{\widehat{{\cal P}}}[Z])]^{\prime}\widehat{\gamma}\\
 &  & +cov_{\widehat{{\cal P}}}(A,Z)^{-1}E_{\widehat{{\cal P}}}\Big[\frac{\partial}{\partial\pi}r(X,\pi_{\widehat{{\cal P}}}(X,Z))(A-\pi_{\widehat{{\cal P}}}(X,Z))(Z-E_{\widehat{{\cal P}}}[Z])\Big]^{\prime}\widehat{\gamma}\\
 &  & -cov_{\widehat{{\cal P}}}(A,Z)^{-1}E_{\widehat{{\cal P}}}[r(X,\pi_{\widehat{{\cal P}}}(X,Z))(Z-E_{\widehat{{\cal P}}}[Z])]^{\prime}\Omega_{\widehat{{\cal P}}}^{-1}\Gamma_{\widehat{{\cal P}}}\widehat{\gamma}
\end{eqnarray*}
which is the expression for the efficient IV estimand relating to
Equation (6).} \\
~

\textsc{(ii) ATE estimand}

{\small{}Analogously to the derivation of the IV estimand above, this
time we substitute the MTE expression from our model into ATE expression
in Table I of \citet[p.1595]{Mogstad2018a} to get $\theta_{ATE}({\cal P})=E_{{\cal P}}[r_{ATE}(X)]^{\prime}\gamma({\cal P})$
as defined in Equation (2). }{\small\par}

{\small{}We find the efficient influence function of $\theta_{ATE}({\cal P})$
as
\begin{eqnarray*}
\frac{\partial}{\partial t}\theta_{ATE}({\cal P}_{t})\Big\vert_{t=0} & = & \Big(\frac{\partial}{\partial t}E_{{\cal P}_{t}}[r_{ATE}(X)]\Big\vert_{t=0}\Big)^{\prime}\gamma({\cal P})+E_{{\cal P}}[r_{ATE}(X)]^{\prime}\frac{\partial}{\partial t}\gamma({\cal P}_{t})\Big\vert_{t=0}\\
 & = & (r_{ATE}(\tilde{x})-E_{{\cal P}}[r_{ATE}(X)])^{\prime}\gamma({\cal P})+E_{{\cal P}}[r_{ATE}(X)]^{\prime}\Omega_{{\cal P}}^{-1}r(\widetilde{x},\pi_{{\cal P}}(\widetilde{x},\widetilde{z}))\big(\tilde{y}-r(\widetilde{x},\pi_{{\cal P}}(\widetilde{x},\widetilde{z}))^{\prime}\gamma({\cal P})\big)\\
 &  & -E_{{\cal P}}[r_{ATE}(X)]^{\prime}\Omega_{{\cal P}}^{-1}\frac{\partial}{\partial\pi}r(\widetilde{x},\pi_{{\cal P}}(\widetilde{x},\widetilde{z}))(\widetilde{a}-\pi_{{\cal P}}(\widetilde{x},\widetilde{z}))r(\widetilde{x},\pi_{{\cal P}}(\widetilde{x},\widetilde{z}))^{\prime}\gamma({\cal P}).
\end{eqnarray*}
Setting the empirical mean of the efficient influence function to
zero, we find an efficient estimator as 
\[
\theta_{ATE}(\widehat{{\cal P}})=E_{\widehat{{\cal P}}}[r_{ATE}(X)]^{\prime}\widehat{\gamma}-E_{\widehat{{\cal P}}}[r_{ATE}(X)]^{\prime}\Omega_{\widehat{{\cal P}}}^{-1}\Gamma_{\widehat{{\cal P}}}\widehat{\gamma}.
\]
}\\
{\small{}~}{\small\par}

\textsc{(iii) ATT estimand}

{\small{}Again we substitute the MTE expression from our model into
the ATT expression in Table II of \citet[p.1595]{Mogstad2018a} to
get $\theta_{ATT}({\cal P})={\cal P}(A=1)^{-1}E_{{\cal P}}[r_{ATT}(X,Z,{\cal P})]^{\prime}\gamma({\cal P})$
as defined in Equation (3). }{\small\par}

{\small{}We find the efficient influence function of $\theta_{ATT}({\cal P})$
as 
\begin{eqnarray*}
\frac{\partial}{\partial t}\theta_{ATT}({\cal P}_{t})\Big\vert_{t=0} & = & \frac{\partial}{\partial t}{\cal P}_{t}(A=1)^{-1}\Big\vert_{t=0}\cdot E_{{\cal P}}[r_{ATT}(X,Z,{\cal P})]^{\prime}\gamma({\cal P})+{\cal P}(A=1)^{-1}\Big(\frac{\partial}{\partial t}E_{{\cal P}_{t}}[r_{ATT}(X,Z,{\cal P}_{t})]\Big\vert_{t=0}\Big)^{\prime}\gamma({\cal P})\\
 &  & +{\cal P}(A=1)^{-1}E_{{\cal P}}[r_{ATT}(X,Z,{\cal P})]^{\prime}\frac{\partial}{\partial t}\gamma({\cal P}_{t})\Big\vert_{t=0}.
\end{eqnarray*}
We deal with each of the three terms on the RHS of this expression
in turn. For the first term, note that 
\[
\frac{\partial}{\partial t}{\cal P}_{t}(A=1)^{-1}\Big\vert_{t=0}=-P(A=1)^{-2}({\cal I}_{\tilde{a}}(1)-P(A=1))
\]
so that 
\[
\frac{\partial}{\partial t}{\cal P}_{t}(A=1)^{-1}\Big\vert_{t=0}\cdot E_{{\cal P}}[r_{ATT}(X,Z,{\cal P})]^{\prime}\gamma({\cal P})=\theta_{ATT}({\cal P})-\theta_{ATT}({\cal P})P(A=1)^{-1}{\cal I}_{\tilde{a}}(1).
\]
}{\small\par}

{\small{}Next, for the second term, 
\begin{eqnarray*}
\frac{\partial}{\partial t}E_{{\cal P}_{t}}[r_{ATT}(X,Z,{\cal P}_{t})]\Big\vert_{t=0} & = & r_{ATT}(\tilde{x},\pi_{{\cal P}}(\tilde{x},\tilde{z}))-E_{{\cal P}}[r_{ATT}(X,\pi_{{\cal P}}(X,Z))]\\
 &  & +\frac{\partial}{\partial\pi}r_{ATT}(\tilde{x},\pi_{{\cal P}}(\tilde{x},\tilde{z}))({\cal I}_{\tilde{a}}(1)-\pi_{{\cal P}}(\tilde{x},\tilde{z}))
\end{eqnarray*}
by Lemma 1. Hence, 
\begin{eqnarray*}
{\cal P}(A=1)^{-1}\Big(\frac{\partial}{\partial t}E_{{\cal P}_{t}}[r_{ATT}(X,Z,{\cal P}_{t})]\Big\vert_{t=0}\Big)^{\prime}\gamma({\cal P}) & = & {\cal P}(A=1)^{-1}r_{ATT}(\tilde{x},\pi_{{\cal P}}(\tilde{x},\tilde{z}))^{\prime}\gamma({\cal P})-\theta_{ATT}({\cal P})\\
 &  & +{\cal P}(A=1)^{-1}\Big(\frac{\partial}{\partial\pi}r_{ATT}(\tilde{x},\pi_{{\cal P}}(\tilde{x},\tilde{z}))({\cal I}_{\tilde{a}}(1)-\pi_{{\cal P}}(\tilde{x},\tilde{z}))\Big)^{\prime}\gamma({\cal P}).
\end{eqnarray*}
Using Lemma 6, the third term is equal to 
\begin{align*}
{\cal P}(A=1)^{-1}E_{{\cal P}}[r_{ATT}(X,Z,{\cal P})]^{\prime}\Omega_{{\cal P}}^{-1}r(\widetilde{x},\pi_{{\cal P}}(\widetilde{x},\widetilde{z}))\big(\widetilde{y}-r(\widetilde{x},\pi_{{\cal P}}(\widetilde{x},\widetilde{z}))^{\prime}\gamma({\cal P})\big)\\
-{\cal P}(A=1)^{-1}E_{{\cal P}}[r_{ATT}(X,Z,{\cal P})]^{\prime}\Omega_{{\cal P}}^{-1}\frac{\partial}{\partial\pi}r(\widetilde{x},\pi_{{\cal P}}(\widetilde{x},\widetilde{z}))(\widetilde{a}-\pi_{{\cal P}}(\widetilde{x},\widetilde{z}))r(\widetilde{x},\pi_{{\cal P}}(\widetilde{x},\widetilde{z}))^{\prime}\gamma({\cal P}).
\end{align*}
Combining the above results and setting the mean of the efficient
influence function to zero, we find an efficient estimator as 
\begin{eqnarray*}
\theta_{ATT}({\cal P}) & = & \widehat{{\cal P}}(A=1)^{-1}E_{\widehat{{\cal P}}}[r_{ATT}(X,\pi_{\widehat{{\cal P}}}(X,Z))]^{\prime}\widehat{\gamma}\\
 &  & +\widehat{{\cal P}}(A=1)^{-1}E_{\widehat{{\cal P}}}\Big[\frac{\partial}{\partial\pi}r_{ATT}(X,\pi_{{\cal \widehat{P}}}(X,Z))(A-\pi_{\widehat{{\cal P}}}(X,Z))\Big]^{\prime}\widehat{\gamma}\\
 &  & -\widehat{{\cal P}}(A=1)^{-1}E_{\widehat{{\cal P}}}[r_{ATT}(X,Z,\widehat{{\cal P}})]^{\prime}\Omega_{{\cal \widehat{P}}}^{-1}\Gamma_{\widehat{{\cal P}}}\widehat{\gamma}.
\end{eqnarray*}
}{\small\par}

\textsc{(iv) ATU estimand}

{\small{}The efficient ATU estimator can be derived using very similar
arguments as used for the efficient ATT estimator above. }\\
{\small{}~}{\small\par}

\textsc{(v) ASG estimand}

{\small{}Follows immediately from the expressions of the ATT and ATU
estimands.\hfill{}$\square$}{\small\par}

\newpage{}

\subsection*{Additional Simulation Results: Bias Performance and Bandwidth Choice}
\noindent \begin{center}
\includegraphics[height=5cm]{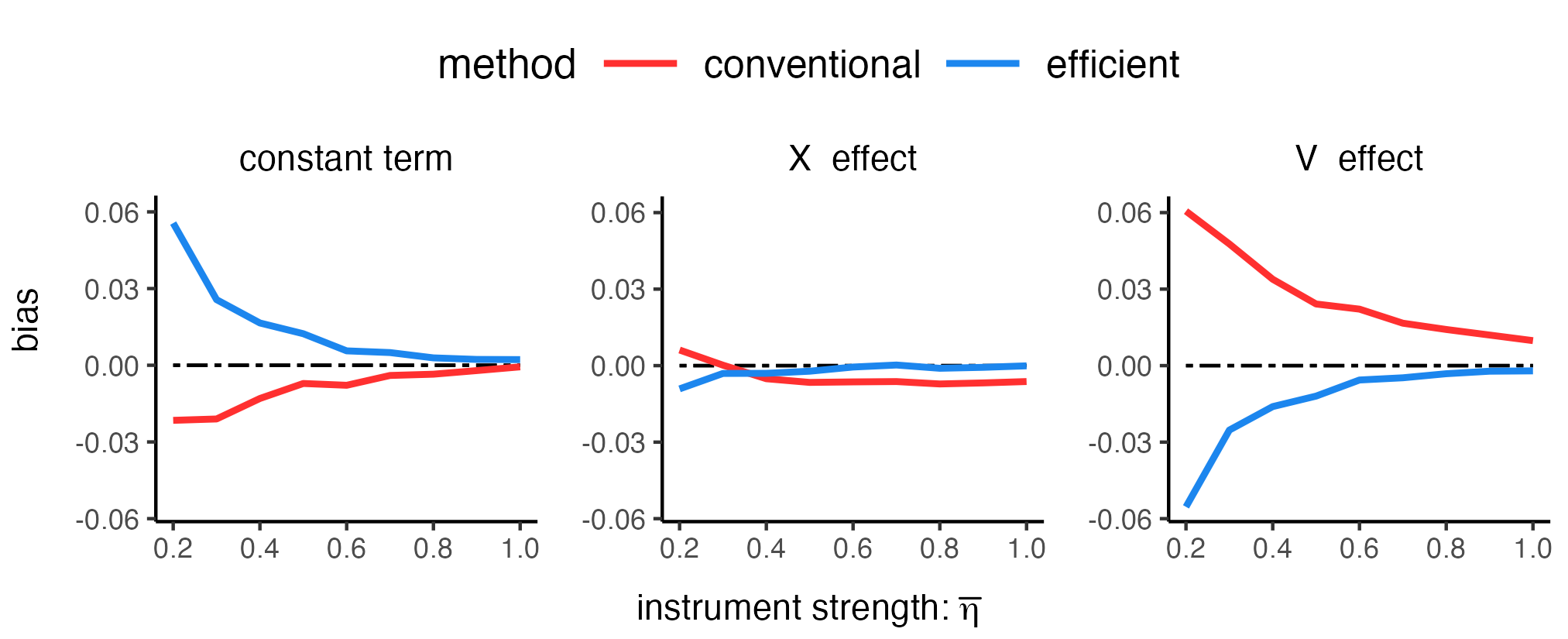}~\\
{\small{}Figure S1. Average bias of MTE parameter estimates over 1000
simulation experiments for varying instrument strength. }{\small\par}
\par\end{center}

~
\noindent \begin{center}
\includegraphics[height=5cm]{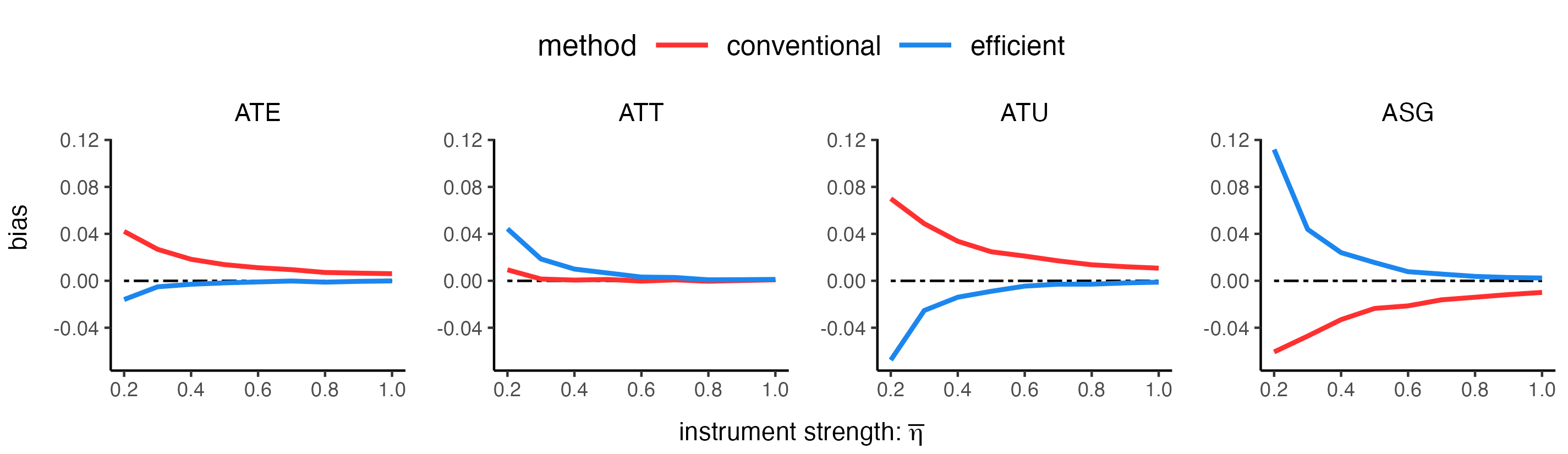}~\\
{\small{}Figure S2. Average bias of target parameter estimates over
1000 simulation experiments for varying instrument strength. }{\small\par}
\par\end{center}

~
\noindent \begin{center}
\includegraphics[height=5cm]{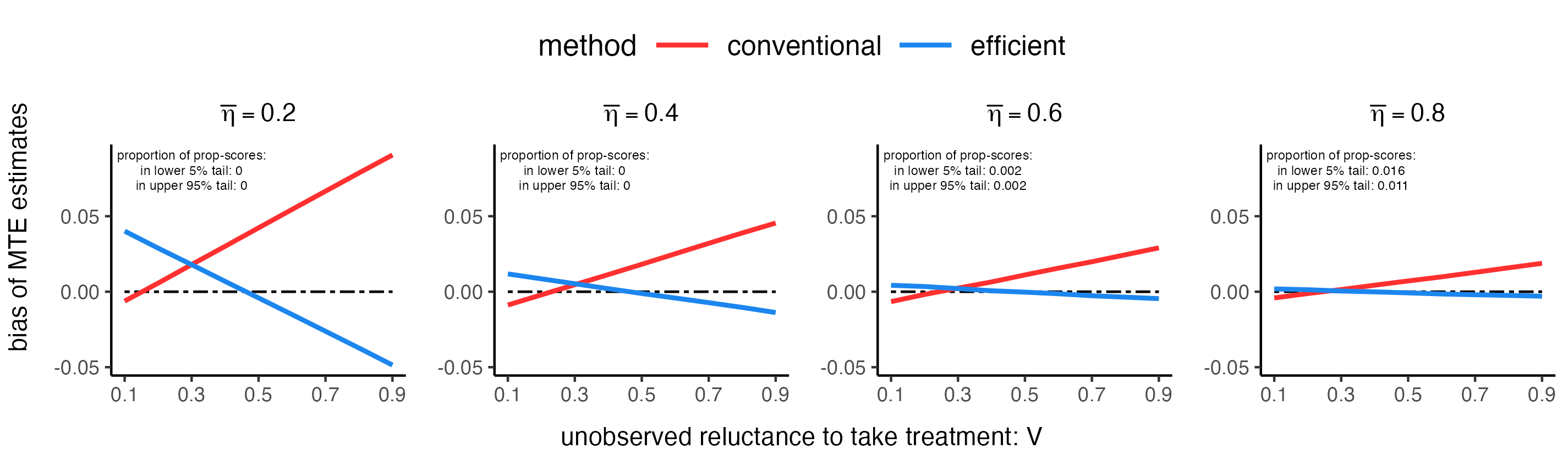}~\\
{\small{}Figure S3. Average bias of MTE estimates averaged over covariates,
$E[Y_{1}-Y_{0}\vert V]$. }{\small\par}
\par\end{center}

\newpage{}

In Figures S4 and S5 below, we considered the sensitivity of estimates
to the choice of bandwidth. Let $h_{0}$ be the averaged bandwidth
using cross-validation on $3$ random samples of size $1,000$ as
discussed in Section 4.2 of the main text. We present the bias performance
of the conventional and efficient estimates as the selected bandwidth
$h=\max\{0.005,h_{0}+0.5\kappa\}$ moved away from $h_{0}$ by varying
the parameter choice $\kappa$ away from zero. \\
~
\noindent \begin{center}
\includegraphics[height=5cm]{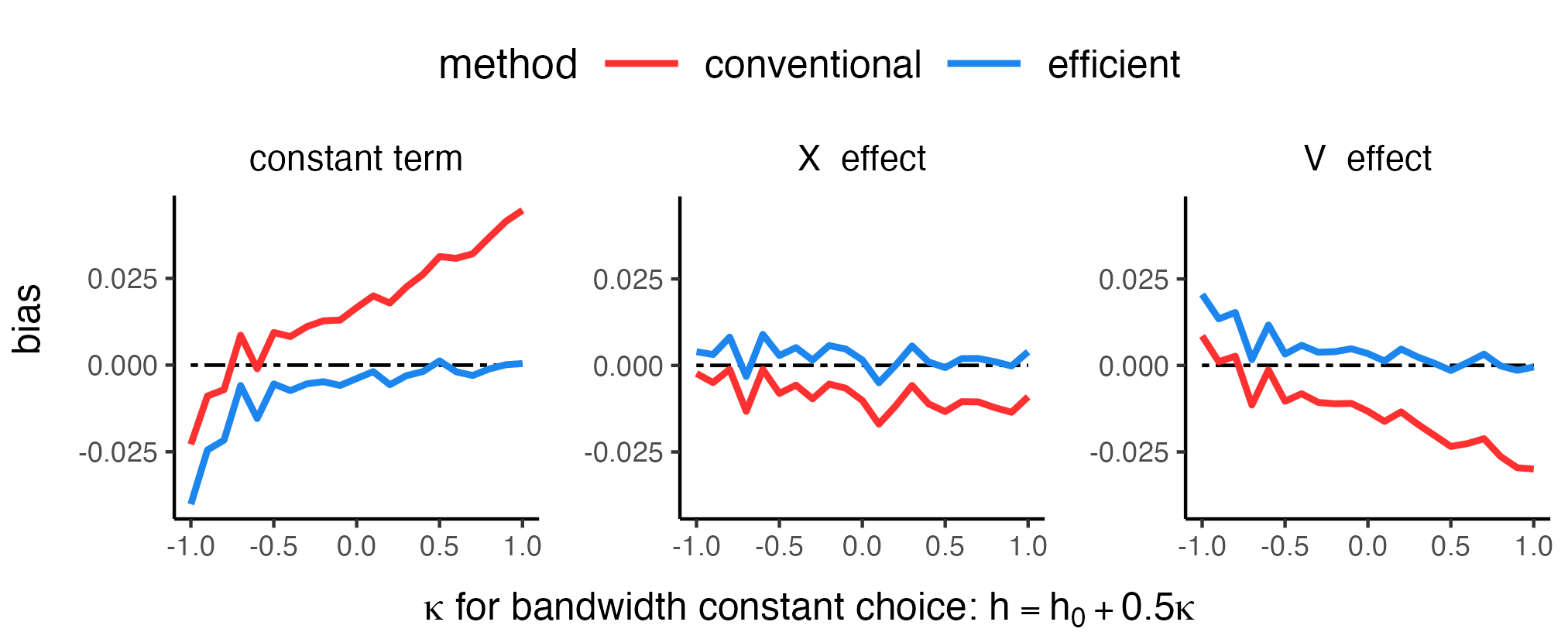}~\\
{\small{}Figure S4. Average bias of MTE parameter estimates over 1000
simulation experiments for bandwidth constant choice $h$. }{\small\par}
\par\end{center}

\noindent \begin{center}
~
\par\end{center}

\noindent \begin{center}
\includegraphics[height=5cm]{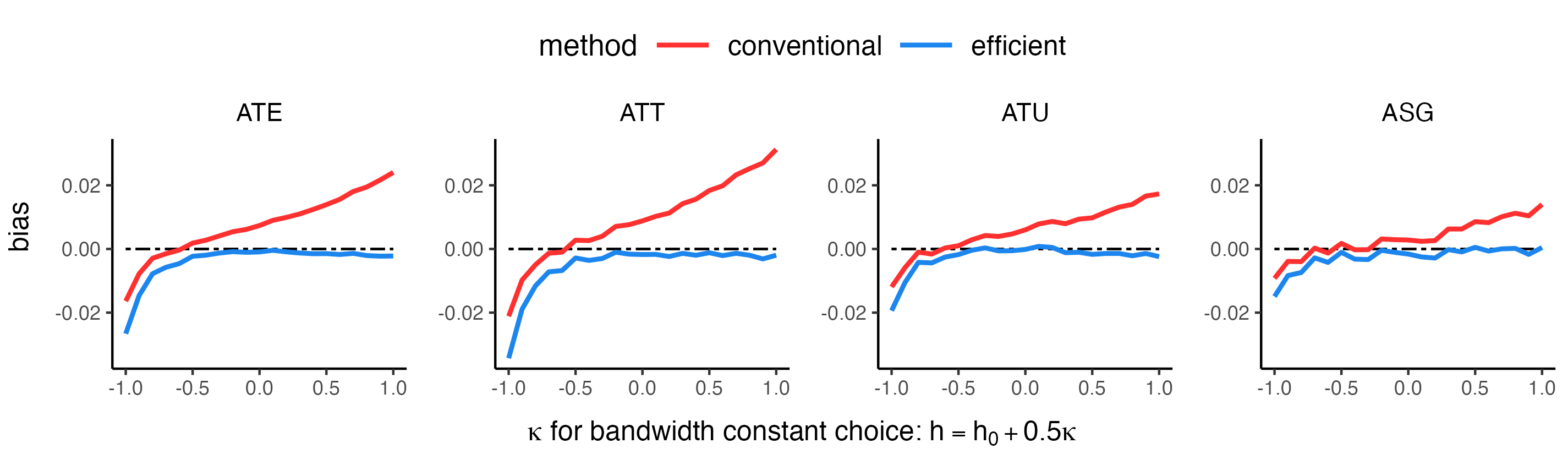}~\\
{\small{}Figure S5. Average bias of target estimand point estimates
over 1000 simulation experiments for bandwidth constant choice $h$. }{\small\par}
\par\end{center}

\newpage{}

\subsection*{Additional Empirical Results: MTE Model with Interaction Effects}
\noindent \begin{center}
\begin{tabular}{c|c|c|c|c}
\hline 
 & \multicolumn{4}{c}{{\small{}Model: MTE with an interaction $X\times V$ effect}}\tabularnewline
\hline 
{\small{}MTE Parameter} & {\small{}constant term} & {\small{}woman $(X)$} & \multicolumn{1}{c|}{{\small{}health consciousness $(V)$}} & \multicolumn{1}{c}{{\small{}interaction effect $(X\times V)$}}\tabularnewline
\hline 
{\small{}Estimate} & {\small{}2.454} & {\small{}-1.797} & {\small{}-3.032} & {\small{}2.001}\tabularnewline
{\small{}Standard Error} & {\small{}0.896} & {\small{}1.123} & {\small{}1.329} & {\small{}2.266}\tabularnewline
{\small{}95\% CI-Lower} & {\small{}0.697} & {\small{}-3.998} & {\small{}-5.638} & {\small{}-2.440}\tabularnewline
{\small{}95\% CI-Upper} & {\small{}4.211} & {\small{}0.404} & {\small{}-0.427} & {\small{}6.441}\tabularnewline
\hline 
\end{tabular}\\
{\small{}~}\\
{\small{}Table S1. Efficient estimates of MTE parameters in a model
with interaction effects, $\theta_{\text{MTE}}(X,V)\sim\text{constant term}+X_{\text{woman}}+V_{\text{health-consciousness}}+(X_{\text{woman}}\times V_{\text{health-consciousness}})$.}\\
{\small{}~}{\small\par}
\par\end{center}

\noindent \begin{center}
\begin{tabular}{c|c|c|c|c}
\hline 
 & \multicolumn{4}{c}{{\small{}Model: MTE with an interaction $X\times V$ effect}}\tabularnewline
\hline 
{\small{}Estimand} & {\small{}ATE} & {\small{}ATT} & \multicolumn{1}{c|}{{\small{}ATU}} & \multicolumn{1}{c}{{\small{}ASG}}\tabularnewline
\hline 
{\small{}Estimate} & {\small{}0.525} & {\small{}1.044} & {\small{}-0.073} & {\small{}1.117}\tabularnewline
{\small{}Standard Error} & {\small{}0.202} & {\small{}0.351} & {\small{}0.482} & {\small{}0.730}\tabularnewline
{\small{}95\% CI-Lower} & {\small{}0.128} & {\small{}0.356} & {\small{}-1.018} & {\small{}-0.314}\tabularnewline
{\small{}95\% CI-Upper} & {\small{}0.921} & {\small{}1.731} & {\small{}0.871} & {\small{}2.549}\tabularnewline
\hline 
\end{tabular}\\
{\small{}~}\\
{\small{}Table S2. Efficient estimates of target estimands in a model
with interaction effects, $\theta_{\text{MTE}}(X,V)\sim\text{constant term}+X_{\text{woman}}+V_{\text{health-consciousness}}+(X_{\text{woman}}\times V_{\text{health-consciousness}})$.}\\
{\small{}~}{\small\par}
\par\end{center}

\noindent \begin{center}
\includegraphics[height=5.5cm]{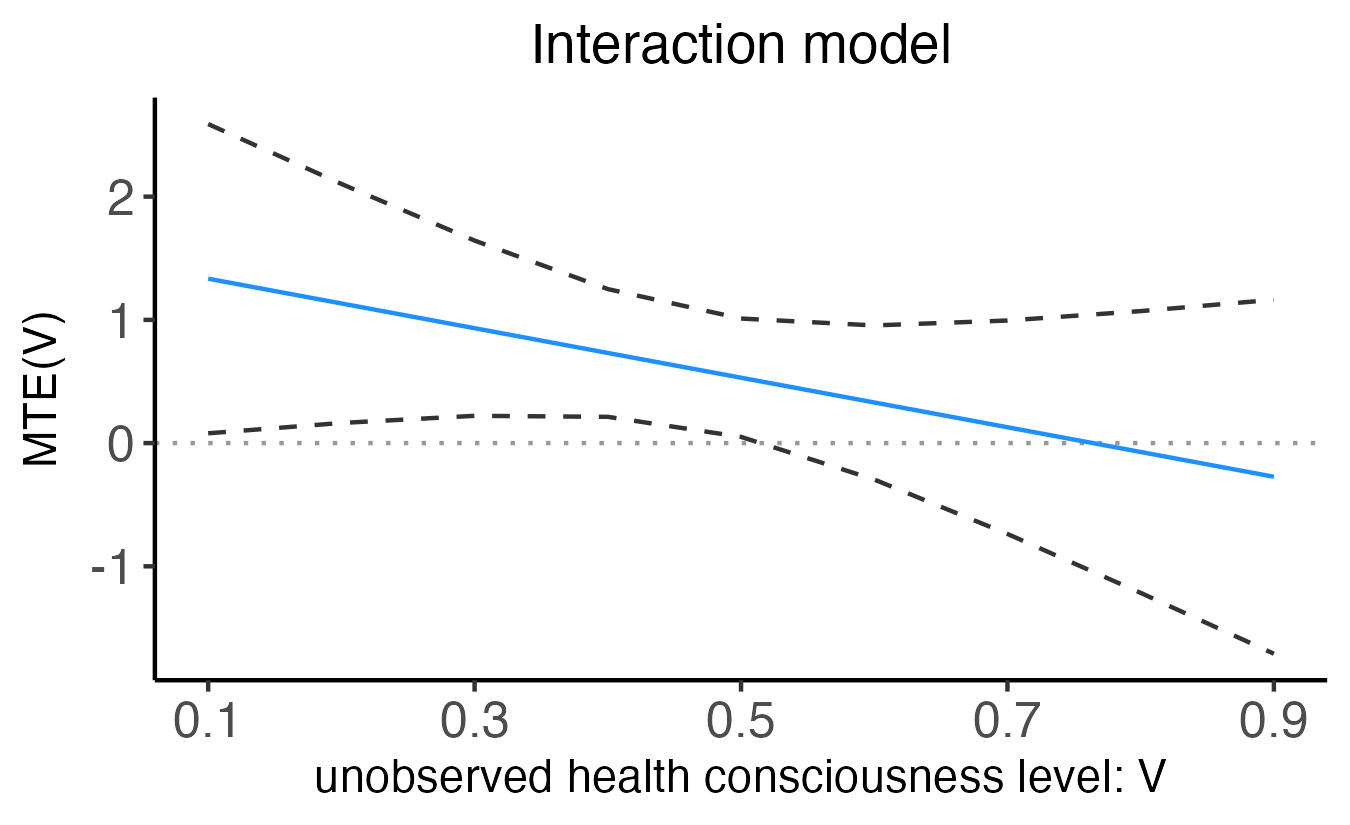}~\\
{\small{}Figure S6. Efficient MTE estimates averaged over covariates,
$E[Y_{1}-Y_{0}\vert V]$ (solid blue line) under a model that allows
the MTE slope in $V$ to differ for men and women, however no sex
difference in the MTE slope in $V$ was estimated (see Table S1).
The dashed black curves indicate 95\% asymptotic confidence intervals. }{\small\par}
\par\end{center}
\end{document}